\documentclass[11pt,a4paper]{article}

\usepackage[T1]{fontenc}
\usepackage[utf8]{inputenc}
\usepackage[a4paper, margin=2.5cm]{geometry}
\PassOptionsToPackage{hyphens}{url}

\usepackage{amsmath,amssymb,amsfonts}
\usepackage{amsthm}
\usepackage{graphicx}
\usepackage{booktabs}
\usepackage{multirow}
\usepackage{xcolor}
\usepackage{hyperref}
\hypersetup{colorlinks=true,linkcolor=blue!70!black,citecolor=blue!70!black,urlcolor=blue!70!black}
\usepackage{url}
\usepackage{microtype}
\usepackage{enumitem}
\usepackage{tikz}
\usetikzlibrary{shapes.geometric, arrows.meta, positioning, fit, backgrounds, calc}
\usepackage{pgfplots}
\pgfplotsset{compat=1.18}
\usepackage{tabularx}
\usepackage{array}
\usepackage{caption}
\usepackage{subcaption}
\usepackage{tcolorbox}
\tcbuselibrary{skins,breakable}
\usepackage[numbers,sort&compress]{natbib}
\usepackage{authblk}

\definecolor{alertred}{RGB}{200,40,40}
\definecolor{defblue}{RGB}{30,100,170}
\definecolor{defbg}{RGB}{235,245,255}
\definecolor{examplebg}{RGB}{255,248,235}
\definecolor{exampleborder}{RGB}{200,150,50}
\definecolor{gapbg}{RGB}{255,240,240}
\definecolor{gapborder}{RGB}{200,50,50}

\newtcolorbox{definitionbox}[1]{
  colback=defbg, colframe=defblue, fonttitle=\bfseries\small,
  title={Definition: #1}, left=4pt, right=4pt, top=3pt, bottom=3pt,
  boxrule=0.6pt, arc=3pt
}
\newtcolorbox{casebox}[1]{
  colback=examplebg, colframe=exampleborder, fonttitle=\bfseries\small,
  title={Case Study: #1}, left=4pt, right=4pt, top=3pt, bottom=3pt,
  boxrule=0.6pt, arc=3pt
}
\newtcolorbox{gapbox}[1]{
  colback=gapbg, colframe=gapborder, fonttitle=\bfseries\small,
  title={Research Gap \##1}, left=4pt, right=4pt, top=3pt, bottom=3pt,
  boxrule=0.6pt, arc=3pt
}
\definecolor{propbg}{RGB}{240,255,240}
\definecolor{propborder}{RGB}{40,140,40}
\newtcolorbox{propositionbox}[1]{
  colback=propbg, colframe=propborder, fonttitle=\bfseries\small,
  title={#1}, left=4pt, right=4pt, top=3pt, bottom=3pt,
  boxrule=0.6pt, arc=3pt
}

\definecolor{takebg}{RGB}{250,247,232}
\definecolor{takeborder}{RGB}{170,135,40}
\newtcolorbox{takeawaybox}[1]{
  colback=takebg, colframe=takeborder, fonttitle=\bfseries\small,
  title={Take-away #1}, left=4pt, right=4pt, top=3pt, bottom=3pt,
  boxrule=0.5pt, arc=2pt, breakable
}

\newtheorem{theorem}{Theorem}

\newcommand{\secref}[1]{Section~\ref{#1}}
\newcommand{\figref}[1]{Fig.~\ref{#1}}
\newcommand{\tabref}[1]{Table~\ref{#1}}
\newcommand{\lasm}{\textsc{LASM}}
\newcommand{\eg}{\emph{e.g.,}}
\newcommand{\ie}{\emph{i.e.,}}
\newcommand{\etal}{\emph{et al.}}
\newcommand{\tmtag}[4]{%
  \par\noindent\hspace*{0pt}%
  \begin{tcolorbox}[
    colback=gray!7, colframe=gray!50, boxrule=0.4pt, arc=2pt,
    left=4pt, right=4pt, top=2pt, bottom=2pt,
    enhanced, before skip=2pt, after skip=4pt
  ]
  {\scriptsize\itshape Threat model:}~{\scriptsize Position=\textbf{#1}; Knowledge=\textbf{#2}; Goal=\textbf{#3}; Persistence=\textbf{#4}.}
  \end{tcolorbox}
}

\begin{document}

\title{\textbf{A Systematic Survey of Security Threats and Defenses in LLM-Based AI Agents:\\
A Layered Attack Surface Framework}}

\author[]{Kexin Chu\thanks{Corresponding author: \texttt{kexin.chu@uconn.edu}}}
\affil[]{School of Computing, University of Connecticut}

\date{}


\maketitle

\begin{abstract}
Agentic AI systems plan across multiple sessions, retain memory, invoke external tools, and coordinate with peer agents. Stateless LLMs do none of these. Existing security taxonomies sort threats by attack type, such as prompt injection or jailbreaking, but say nothing about which architectural component is vulnerable or over what timescale the threat manifests. This paper addresses the structural questions directly. We propose the Layered Attack Surface Model (\lasm), a seven-layer decomposition of the agentic stack: Foundation, Cognitive, Memory, Tool Execution, Multi-Agent Coordination, Ecosystem, and Governance. A non-transferability theorem (sketch in §IV, full proof in supplementary) shows that a control at one layer has zero detection power against an attack localized at another. We add an orthogonal temporality axis with four classes: instantaneous (T1), session-persistent (T2), cross-session cumulative (T3), and sub-session-stack (T4). Plotting 116 papers (2021--2026) on the resulting 7$\times$4 grid produces four numbered findings. The under-studied zone $\{L_5,L_6,L_7\} \times \{T_3,T_4\}$ holds 6.3\% of all 144 paper-cell assignments yet contains the highest-severity threats (EF1). Every surveyed L4 attack reduces to one cause, principal trust inversion (EF2). Seven of 28 grid cells have zero defense coverage, and three of those seven contain documented attacks (EF3). No surveyed benchmark evaluates T3 or T4 threats, so progress on the right-hand columns is unmeasurable (EF4). Beyond these findings, we provide a cross-layer defense taxonomy with prescriptive recipes for seven canonical attack classes, and a dependency DAG that separates engineering-tractable gaps from fundamental open problems. The per-paper coding, a robustness analysis script for EF1 (perturbations keep the fraction below 9\%), and a reference ABOM JSON Schema with validator and 12 unit tests are released as supplementary material.
\end{abstract}


\section{Introduction}
\label{sec:intro}

The agentic AI threat surface is not the LLM threat surface scaled up. A stateless LLM consumes one input and emits one output, so input/output filtering covers it. An AI agent maintains persistent state across sessions, formulates multi-step plans, invokes privileged tools such as web browsers, code interpreters, file systems and payment APIs~\cite{toolformer2023,toolllm2024,gorilla2024,shenHuggingGPT2023}, and increasingly coordinates with peer agents on long-horizon tasks~\cite{react2023,voyager2023,autogpt2023,agentsurvey2024,riseAgents2023}. The capability comes from emergent properties at scale~\cite{wei2022emergent} and from tool and memory infrastructure~\cite{nakano2022webgpt}. The security cost is that every architectural decision that makes an agent more capable than a chatbot enlarges the surface adversaries can exploit.

Three properties of this attack surface have no analogue for stateless LLMs. Attacks are \emph{emergent}: harm arises from the composition of authorized actions rather than from any single action. Attacks are \emph{compositional}: a defense placed at one component cannot detect an attack routed through another. Attacks are \emph{temporally extended}: a payload can be installed weeks before it executes, and may never execute as a discrete event. Three production incidents make the gap concrete.

\begin{itemize}[leftmargin=1.5em,noitemsep]
  \item An adversarially crafted document retrieved during a routine web search can silently corrupt an agent's long-term memory, influencing behavior in a completely unrelated session weeks later, with no detectable anomaly at either the injection or exploitation point.
  \item A compromised sub-agent in an enterprise multi-agent pipeline can propagate malicious intent through peer-to-peer trust relationships, compromising the entire orchestration network from a single point of entry~\cite{infectiousJailbreak2024}.
  \item A malicious MCP server, indistinguishable to the agent from a legitimate tool provider, can silently exfiltrate sensitive data via hidden instructions embedded in tool descriptions~\cite{mcp2025}.
\end{itemize}

These threats have no direct analogues in classical LLM safety or traditional software security. The security community has responded by extending existing frameworks rather than building new ones.

\paragraph{Research Questions.}
This survey is structured as a systematization of knowledge (SoK) and is organized around five research questions (RQs). The questions follow the SoK convention of stating in advance which knowledge gaps the systematization is intended to close.

\begin{itemize}[leftmargin=2em,noitemsep,label=]
\item \textbf{RQ1 (Architectural decomposition)}: Can the agentic attack surface be decomposed into a small number of architectural layers with disjoint trust boundaries and disjoint defense leverage points, such that defense placement decisions follow from the decomposition?
\item \textbf{RQ2 (Temporal structure)}: Do agentic attacks exhibit a temporal structure that is invisible to type-centric taxonomies, and if so, what is the minimal classification needed to capture this structure?
\item \textbf{RQ3 (Coverage gaps)}: Where in the (architecture, temporality) space is research effort concentrated, and where is it absent? Is severity correlated, uncorrelated, or inversely correlated with effort?
\item \textbf{RQ4 (Defense coverage)}: Does the defense literature cover the same (architecture, temporality) cells as the attack literature, and where does each cell have zero defense coverage?
\item \textbf{RQ5 (Operationalization)}: Are the gaps tractable as engineering problems, or do they require fundamental research?
\end{itemize}

These RQs are answered respectively by \secref{sec:framework} (RQ1, via \lasm); \secref{sec:framework}\,{\S}IV.B (RQ2, via temporality T1--T4); Empirical Findings~1 and~2 in \secref{sec:framework} and \secref{sec:tool} (RQ3, via the coverage heatmap and L4 root-cause analysis); Empirical Findings~3 and~4 in \secref{sec:defcov} and \secref{sec:eval} (RQ4, via the defense coverage heatmap and benchmark-vacuum claim); and \secref{sec:gaps} together with \secref{sec:future} (RQ5, via the gap-to-tradition mapping in \tabref{tab:future_map}).

\paragraph{Limitations of Existing Work.}

Recent surveys \cite{agentsUnderThreat2025,agentSecSurvey2025a,agentSecSurvey2025b} have catalogued an expanding roster of attacks. However, they share structural limitations that reduce their actionability for system designers:

\textbf{(1) Type-centric, not component-centric taxonomies.} Organizing threats by attack type (jailbreaking, injection, poisoning) conflates attacks that require different defenses at different layers. A gradient-based adversarial perturbation at the embedding level and a memory poisoning attack via a RAG document are both classified as ``adversarial inputs'', yet one is mitigated at the model layer and the other at the memory management layer. A designer cannot derive \emph{where to place security controls} from a type-centric taxonomy.

\textbf{(2) No temporal dimension.} Existing surveys treat attacks as instantaneous events. In agentic systems, the most dangerous attacks are temporally extended: a payload injected today may execute weeks later via long-term memory, or may never trigger an explicit ``attack event'' at all (emergent misalignment). This temporal structure is invisible to existing taxonomies.

\textbf{(3) Coverage gap at higher layers.} The literature heavily covers L1 (foundation model jailbreaking) and L4 (prompt injection), but higher-layer threats, including multi-agent collusion, ecosystem supply chain attacks, and governance accountability, are severely under-studied both empirically and theoretically.

\textbf{(4) Dated coverage.} Several surveys have cutoffs in 2023--2024, predating MCP security vulnerabilities, steganographic collusion benchmarks, zero-day autonomous exploitation by agent teams, and agentic governance toolkits.

\paragraph{Contributions.}

This paper makes five contributions, each anchored to a research question (\secref{sec:methodology}):

\begin{enumerate}[leftmargin=1.5em,noitemsep]
  \item \textbf{\lasm} (\secref{sec:framework}): a seven-layer decomposition of the agentic attack surface with disjoint trust boundaries and disjoint defense leverage points, formalized in Theorem~\ref{thm:nontrans} (Appendix~\ref{app:proof}). Unlike OSI-inspired network models, \lasm\ accounts for a property unique to LLM-based agents: information crosses layer boundaries in semantically opaque ways, enabling cross-layer attacks invisible to per-layer detectors.

  \item \textbf{Attack temporality} (\secref{sec:temporality}): an orthogonal four-class dimension (T1--T4) that captures \emph{when} attacks manifest. Plotting the corpus on \lasm\ $\times$ temporality reveals that the highest-impact threats are not the most instantaneous but the slowest, a structure no type-centric taxonomy can express.

  \item \textbf{Four numbered Empirical Findings} (\secref{sec:framework}, \secref{sec:tool}, \secref{sec:defcov}, \secref{sec:eval}) extracted from systematic review of 116 papers (2021--2026): an inverse severity--effort correlation in the under-studied zone, root-cause uniformity at L4, defense--attack coverage asymmetry, and a complete absence of T3/T4 benchmarks.

  \item \textbf{Cross-layer defense taxonomy with prescriptive recipes} (\secref{sec:defense}, \secref{sec:recipes}): 15 defense classes mapped onto the \lasm\ $\times$ temporality grid, plus seven canonical attack-class recipes specifying joint defense stacks and residual risks.

  \item \textbf{Gap analysis with dependency structure} (\secref{sec:gaps}, \secref{sec:future}): five research gaps in \tabref{tab:gaps} positioned by prerequisite structure, separating tractable engineering work from fundamental open problems.
\end{enumerate}

\paragraph{Scope.}

We cover security challenges arising from or enabled by LLM-based agentic AI architecture. We explicitly exclude: (a) static LLM safety (bias, toxicity) for non-agentic models; (b) classical adversarial ML against vision/tabular models; (c) AI misuse by humans (malware authoring, phishing automation), which is addressed by a separate literature; and (d) physical security of hardware infrastructure. We include autonomous vulnerability exploitation (\secref{sec:tool}) because it characterizes the \emph{offensive capability surface} that agentic architecture introduces: understanding this capability is prerequisite to designing appropriate authorization and containment controls.

\paragraph{Paper Organization.}

\figref{fig:outline} renders the dependency structure of the survey. \secref{sec:methodology} describes the survey methodology. \secref{sec:background} reviews agent architectures and the formal threat model. \secref{sec:framework} introduces \lasm\ and attack temporality and is the analytical backbone for the layer-specific surveys in \secref{sec:L1L2}--\secref{sec:governance}. \secref{sec:defense} presents the cross-layer defense taxonomy; \secref{sec:eval} surveys evaluation benchmarks; \secref{sec:gaps} identifies the five research gaps that \secref{sec:future} maps to research directions. \secref{sec:related} compares with prior surveys. \secref{sec:conclusion} concludes. Appendix~\ref{app:coding} contains the operational coding rules and a sensitivity analysis underlying \tabref{tab:heatmap_counts}.

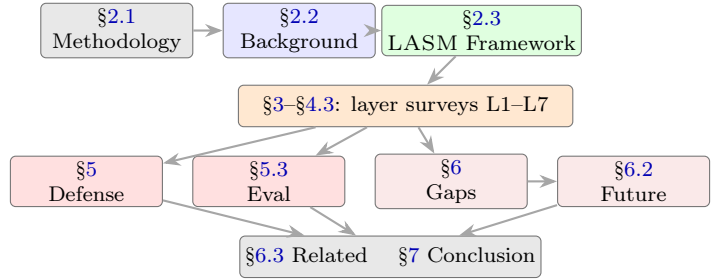
\begin{figure}[h]
\centering
\begin{tikzpicture}[
  font=\scriptsize,
  node/.style={draw=black!55, rounded corners=2pt, minimum width=2.0cm, minimum height=0.55cm, align=center, inner sep=2pt},
  back/.style={fill=blue!10}, fwk/.style={fill=green!12}, lyr/.style={fill=orange!18},
  syn/.style={fill=red!12}, gap/.style={fill=alertred!10}, meta/.style={fill=gray!18},
  arr/.style={-Stealth, thick, gray!70}
]
  \node[node, meta] (m) at (0,0)        {\S\ref{sec:methodology}\\Methodology};
  \node[node, back] (b) at (2.4,0)      {\S\ref{sec:background}\\Background};
  \node[node, fwk]  (f) at (4.8,0)      {\S\ref{sec:framework}\\\lasm\ Framework};
  \node[node, lyr, minimum width=4.4cm] (L) at (3.8,-1.0)
        {\S\ref{sec:L1L2}--\S\ref{sec:governance}: layer surveys L1--L7};
  \node[node, syn]  (d)  at (-0.4,-2.0) {\S\ref{sec:defense}\\Defense};
  \node[node, syn]  (e)  at (2.0,-2.0)  {\S\ref{sec:eval}\\Eval};
  \node[node, gap]  (g)  at (4.4,-2.0)  {\S\ref{sec:gaps}\\Gaps};
  \node[node, gap]  (fu) at (6.8,-2.0)  {\S\ref{sec:future}\\Future};
  \node[node, meta] (r)  at (3.6,-3.0)  {\S\ref{sec:related} Related~~~~\S\ref{sec:conclusion} Conclusion};
  \draw[arr] (m) -- (b);
  \draw[arr] (b) -- (f);
  \draw[arr] (f) -- (L);
  \draw[arr] (L) -- (d);
  \draw[arr] (L) -- (e);
  \draw[arr] (L) -- (g);
  \draw[arr] (g) -- (fu);
  \draw[arr] (d) -- (r);
  \draw[arr] (e) -- (r);
  \draw[arr] (fu) -- (r);
\end{tikzpicture}
\caption{Survey outline. Methodology and Background feed the \lasm\ framework (\secref{sec:framework}); the framework drives the seven layer-specific surveys (\secref{sec:L1L2}--\secref{sec:governance}), which in turn feed the synthesis sections (Defense, Eval, Gaps) and the proposed Future Directions. After Papernot~\etal~\cite{papernot2018sok} Fig.~1.}
\label{fig:outline}
\end{figure}

\paragraph{How to read this survey.} The paper is dense; readers should not feel obliged to read it linearly. \tabref{tab:roadmap} provides four audience-specific reading paths, each $\sim$8 pages. The full document is intended as a reference; the paths below pull out the most actionable subset for each reader profile.

\begin{table}[h]
\centering
\footnotesize
\caption{Audience-specific reading roadmaps. Sections in \textbf{bold} are essential for that audience; sections in (parentheses) are supporting context.}
\label{tab:roadmap}
\renewcommand{\arraystretch}{1.1}
\begin{tabularx}{\linewidth}{p{2.0cm}X}
\toprule
\textbf{Audience} & \textbf{Suggested path} \\
\midrule
Agent-system designer & (\secref{sec:background}) $\to$ \textbf{\secref{sec:framework}} $\to$ \textbf{\secref{sec:tool}} $\to$ \textbf{\secref{sec:ecosystem}} $\to$ \textbf{\secref{sec:defense}} $\to$ \secref{sec:actionlog}. Focus on tab:lasm, fig:lasm\_stack, fig:defcov, and the action-log schema. \\
Security researcher & \textbf{\secref{sec:methodology}} $\to$ \textbf{\secref{sec:framework}} $\to$ \secref{sec:L1L2}--\secref{sec:governance} (skim) $\to$ \textbf{\secref{sec:gaps}} $\to$ Appendix~\ref{app:coding}. Focus on the coding methodology, Empirical Finding 1, and the five research gaps. \\
Regulator / policy & (\secref{sec:intro}) $\to$ \secref{sec:tool} (skim) $\to$ \textbf{\secref{sec:governance}} $\to$ \textbf{\secref{sec:actionlog}} $\to$ \textbf{Gaps 4 \& 5}. Focus on the EU AI Act / NIST RMF gap analysis and the action-log schema. \\
Practitioner / SOC analyst & \textbf{Take-aways 5.1--10.1} (inline, one per subsection) $\to$ \textbf{\secref{sec:defense}} $\to$ Appendix~\ref{app:owasp} (\tabref{tab:owasp_compare}) $\to$ \tabref{tab:defense}. Focus on the OWASP/ATLAS$\to$\lasm\ mapping and the defense recipes. \\
\bottomrule
\end{tabularx}
\end{table}

\section{Overview}
\label{sec:overview}

This section establishes the foundations the rest of the paper relies on. \secref{sec:methodology} describes the literature search and screening protocol; \secref{sec:background} fixes the agent architecture and threat model used throughout; \secref{sec:framework} introduces the Layered Attack Surface Model (\lasm) and its orthogonal temporality axis, which together provide the analytical lens for \secref{sec:single_agent}--\secref{sec:defenses_eval}.

\subsection{Survey Methodology}
\label{sec:methodology}

This survey follows a systematic literature review protocol inspired by PRISMA~\cite{prisma2020} guidelines for computing literature reviews. \figref{fig:stats} shows the temporal distribution of retained papers.

\subsubsection{Search Strategy}

Four major databases were searched: IEEE Xplore, ACM Digital Library, arXiv (cs.CR, cs.AI, cs.LG), and Google Scholar. Forward/backward citation tracing was additionally performed from three seed papers \cite{agentsUnderThreat2025,promptInjection2023,react2023}. Search terms combined agent-related terms $\times$ security-related terms:

\noindent\textit{Agent terms:} ``AI agent'', ``LLM agent'', ``autonomous agent'', ``multi-agent'', ``agentic AI'', ``tool-augmented LLM''

\noindent\textit{Security terms:} ``security'', ``attack'', ``adversarial'', ``prompt injection'', ``jailbreak'', ``poisoning'', ``safety'', ``trust'', ``vulnerability''

The search covered January 2021 to April 2026. An initial phase ran from January 2021 to April 2025; a second phase re-ran the same queries for May 2025 to April 2026, yielding a unified corpus without cutoff artifacts.

\subsubsection{Inclusion and Exclusion Criteria}

\textit{Inclusion:} (1a) Peer-reviewed papers at top venues (IEEE S\&P, CCS, USENIX Security, NeurIPS, ICML, ICLR, ACL, EMNLP, ACM CCS, NDSS, AAAI, KDD, DSN) are included regardless of citation count; (1b) arXiv and non-top-tier preprints require $\geq$20 citations for papers through December 2024; for January 2025--April 2026 preprints the citation threshold was not applied due to insufficient publication time, and inclusion required a direct topical contribution to an identified gap area; (2) papers specifically addressing security of LLM-based agent systems; (3) papers presenting attacks, defenses, or evaluation frameworks.

\textit{Exclusion:} (1) Papers on LLM safety without agentic context; (2) classical adversarial ML on non-LLM models; (3) position papers without technical contributions; (4) duplicate publications.

\subsubsection{Screening and Statistics}

\tabref{tab:prisma} reports the screening process. The combined two-phase search retrieved 1,634 records; after deduplication and a title/abstract pass, 313 papers proceeded to full-text review. The final corpus of \textbf{116 papers} comprises: 38 on attacks (L1--L4), 26 on multi-agent and ecosystem attacks (L5--L6), 16 on governance and alignment (L7), 18 on defenses, 13 on benchmarks and evaluation, and 5 comprehensive surveys used for comparison. We acknowledge that this corpus covers the most visible work within these search and inclusion criteria but cannot be exhaustive in a field growing at this rate.

\textbf{Layer and temporality coding.} Each retained paper was independently coded by the first author into one or more $(\text{layer}, \text{temporality})$ cells using the following decision rules. \emph{Layer assignment}: a paper is assigned to $L_i$ if its primary technical contribution (attack, defense, or benchmark) targets the trust boundary or attack surface defined for $L_i$ in \secref{sec:framework}; papers addressing multiple layers are coded in each relevant cell. \emph{Temporality assignment}: a paper is assigned to $T_k$ if the attack or defense it studies manifests in the time window defined for $T_k$ (T1: within one inference call; T2: within one session; T3: across sessions; T4: sub-session-stack or trigger-conditioned). A paper is assigned to $T_3$ only if it explicitly involves cross-session state persistence via a memory-layer write; a paper is assigned to $T_4$ only if the payload resides in model weights or training data (T4b) or if there is no discrete payload at all (T4a). For papers where layer or temporality was ambiguous, we applied a conservative rule: assign to the more restrictive (lower-layer, faster-temporality) category.

\textbf{Coding distribution.} The 144 paper-cell assignments in \tabref{tab:heatmap_counts} come from 116 papers, giving a mean of 1.24 cells per paper. Cell-level totals in \tabref{tab:heatmap_counts} are the primary data; per-cell paper assignments are documented in the supplementary \texttt{corpus\_coding.csv} for independent verification. Papers assigned to multiple cells are those covering cross-layer attack chains or comprehensive defense systems. Inter-rater reliability was not formally assessed; this limitation is discussed in \secref{sec:conclusion}.

\textbf{2025--2026 inclusion threshold.} Of the 44 papers from January 2025--April 2026, 18 met the $\geq$20 citation threshold and were included unconditionally (most already published at named venues). The remaining 26 were admitted under the gap-area exception: all address layers or temporality classes with $\leq$3 papers per cell at the time of inclusion. Excluding these 26 papers does not change the main coverage gap finding; it reduces under-studied cell counts in \tabref{tab:heatmap_counts} by at most 2 per cell.

\textbf{Note on total references.} The bibliography contains \textbf{186 entries}, of which 116 are primary security papers that passed PRISMA screening. The remaining 70 entries are background and foundational references: LLM and agent architectures, classical computer security principles, AI safety and alignment theory, and governance documents. These are cited in the introduction, background (Section~\ref{sec:background}), and defense taxonomy sections to contextualize the survey findings but were not included in the coverage statistics of \tabref{tab:prisma}.

\begin{table}[h]
\centering
\small
\caption{Literature screening process (PRISMA-inspired).}
\label{tab:prisma}
\renewcommand{\arraystretch}{1.1}
\begin{tabularx}{\linewidth}{Xrr}
\toprule
\textbf{Stage} & \textbf{Records} & \textbf{Excl.} \\
\midrule
Retrieved (Phase 1: 2021--Apr 2025) & 1,247 & --- \\
Retrieved (Phase 2: May 2025--Apr 2026) & 387 & --- \\
\quad\footnotesize Combined (4 databases, 2 phases) & \footnotesize 1,634 & \\
After deduplication & 1,121 & 513 \\
After title/abstract screen & 313 & 808 \\
\quad\footnotesize(non-agentic; non-security; position) & & \\
After full-text assessment & 116 & 197 \\
\quad\footnotesize(no contribution; out of scope) & & \\
\midrule
\textbf{Final corpus} & \textbf{116} & \\
\bottomrule
\end{tabularx}
\end{table}

\figref{fig:stats} shows the temporal distribution, confirming an exponential publication growth trend beginning in 2023.

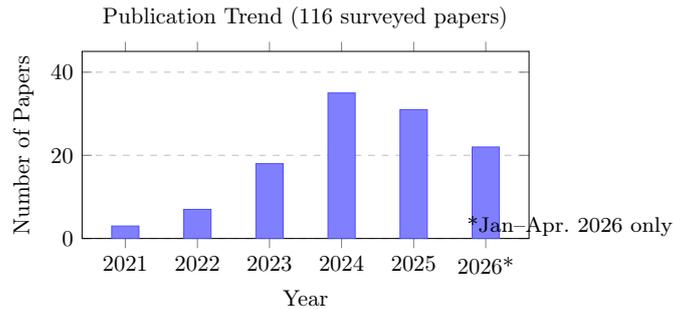
\begin{figure}[t]
\centering
\begin{tikzpicture}[scale=0.85, font=\small]
\begin{axis}[
  ybar, bar width=12pt,
  xlabel={Year}, ylabel={Number of Papers},
  xtick={2021,2022,2023,2024,2025,2026},
  xticklabels={2021,2022,2023,2024,2025,2026*},
  ymin=0, ymax=45,
  width=8.5cm, height=4.5cm,
  ymajorgrids=true, grid style=dashed,
  bar shift=0pt,
  title={Publication Trend (116 surveyed papers)},
  title style={font=\small},
  enlarge x limits=0.12,
]
\addplot[fill=blue!50, draw=blue!70] coordinates {
  (2021,3) (2022,7) (2023,18) (2024,35) (2025,31) (2026,22)
};
\end{axis}
\node[font=\scriptsize, anchor=west] at (5.8,0.2) {*Jan--Apr.\ 2026 only};
\end{tikzpicture}
\caption{Publication trend of surveyed papers by year (116 papers, January 2021--April 2026). The 2025 count (31) reflects the full calendar year; 2026 (22) covers January--April only. The exponential growth trend beginning in 2023 continues into 2026, underscoring the urgency of the coverage gaps identified in \secref{sec:framework}.}
\label{fig:stats}
\end{figure}

\subsubsection{Methodological Caveats}
\label{sec:caveats}

The search has known biases that bound how strongly the empirical claims should be read. The search terms under-sample work that frames an agentic threat without ``agent''-prefixed vocabulary, such as papers titled ``LLM-driven web crawler.'' Forward/backward citation tracing from the three seed papers partly compensates. The 20-citation threshold for non-top-venue work through 2024 excludes recent low-impact preprints, and the sensitivity analysis in Appendix~\ref{app:coding} bounds the effect on the under-studied zone count at $\pm 2$ papers. Publication-time bias works the other way: a top-venue 2024 acceptance is included regardless of citations, while an equivalent arXiv preprint may be excluded. This biases the corpus toward formal venues and against industry-led work in preprints and blog posts.

A single author performed the layer$\times$temporality coding. Appendix~\ref{app:coding} reports operational rules and a sensitivity analysis; formal inter-rater reliability is future work, and no specific cell count is claimed robust beyond $\pm 1$. Two further structural limits: agentic security work in DEF CON or Black Hat talks, vendor advisories, and red-team reports is systematically excluded by the PRISMA-style search (the Anthropic espionage campaign report and the ShadowPrompt Claude browser-extension incident in \secref{sec:framework} are exceptions admitted via citation tracing), and the 18 defense papers are an explicit corpus rather than an enumeration of deployed industry defenses, so \figref{fig:defcov} reports an upper bound on \emph{published} defenses. These caveats bound the precision of the cell counts; they do not change the qualitative findings.

\subsubsection{Coverage Through April 2026: Key New Findings}
\label{sec:postcutoff}

The second search phase (May 2025--April 2026) added 22 primary papers to the corpus. Rather than treating these as supplementary, we incorporate them fully into \tabref{tab:heatmap_counts} and the layer-level analyses. We highlight here the four most structurally significant additions that either sharpen or partially address the Empirical Findings from the 2021--2024 corpus.

\textbf{L1: backdoors become a systematic threat class.} BackdoorLLM~\cite{backdoorLLM2025} (NeurIPS 2025) provides the first comprehensive benchmark covering eight backdoor strategies including chain-of-thought hijacking; NeuroStrike~\cite{neuroStrike2025} (NDSS 2026) demonstrates neuron-level attacks on aligned models that survive RLHF. These venue-confirmed papers increase the L1$\times$T4 cell from 2 to 6, partially offsetting EF1.

\textbf{L3 memory: cross-session attacks gain momentum.} MINJA~\cite{minja2025} (NeurIPS 2025) demonstrates query-only memory injection with 98.2\% injection success; Xu et al.~\cite{memoryControlFlow2026} show that injected memory entries can redirect control flow across tasks. Together with MemoryGraft~\cite{memoryGraft2025}, the L3$\times$T3 cell grows from 3 to 6 papers -- the largest single-cell gain in the extended corpus.

\textbf{L6 MCP ecosystem: from isolated incidents to systematic study.} Seven peer-reviewed papers now address MCP security (Li and Gao at DSN~\cite{mcpEcosystemDSN2026}; Wang et al.\ at AAAI~\cite{mcpTox2025}; Zong et al.\ and Zhang et al.\ at ICLR~\cite{mcpSafetyBench2026,mcpSecBench2026}; Liu et al.~\cite{maliciousSkillsWild2026}; Qu et al.~\cite{supplyChainSkillEco2026}). The L6$\times$T1 cell grows from 4 to 8. Anthropic's disclosure of the first AI-orchestrated espionage campaign~\cite{anthropicAIOrchestrated2025} and the ShadowPrompt browser-extension vulnerability~\cite{shadowPrompt2026} provide real-world confirmation that both L4 and L6 attack surfaces are actively exploited.

\textbf{What still does not change.} Despite the 22 new papers, EF3 (zero-defense cells) and EF4 (T3/T4 benchmark vacuum) hold: no new benchmark evaluates T3 or T4 threats, and the cells $(L_1, T_3)$, $(L_4, T_4)$, $(L_6, T_3)$, and $(L_6, T_4)$ remain defense-free. The under-studied zone $\{L_5$--$L_7\} \times \{T_3, T_4\}$ grows from 8 to 9 assignments -- still 6.3\% of all 144 assignments (EF1). He~\etal~\cite{he2025emerged} (ACM CSUR 2025) and the TrustAgent survey at KDD 2025~\cite{trustworthyLLMAgentsKDD2025} provide parallel frameworks for comparison; \tabref{tab:survey_compare} extends to include both.

\subsection{Anatomy of an AI Agent}
\label{sec:background}

\subsubsection{Core Architectural Components}

\begin{definitionbox}{AI Agent}
An \textbf{AI agent} is an autonomous computational system built upon a foundation language model that: (a) perceives multi-modal inputs from its environment, (b) maintains state across multiple interaction turns, (c) formulates and executes goal-directed plans through sequences of actions, and (d) modifies its environment via external tool invocations or communication with peer agents \cite{agentsurvey2024}.
\end{definitionbox}

We identify six core architectural components (\figref{fig:agent_arch}), each introducing distinct security concerns:

\textbf{(1) Foundation Model.} A large pretrained language model~\cite{vaswani2017attention,devlin2019bert,brown2020language} that serves as the reasoning engine~\cite{gpt4report2023,claudeCard2023,touvron2023llama,touvron2023llama2}. Safety properties acquired through RLHF~\cite{christiano2017rlhf,ouyang2022instructgpt,rlhf2022} or Constitutional AI~\cite{constitutionalAI2022} are necessary but insufficient for agentic deployment, as they are trained on single-turn interaction data~\cite{bommasani2021foundation}.

\textbf{(2) Planning and Reasoning Module.} Translates high-level goals into action sequences via architectures including ReAct~\cite{react2023}, chain-of-thought~\cite{chainOfThought2022}, and tree-of-thought search. Multi-step planning creates a surface for cumulative goal drift: small perturbations at each reasoning step can compound into large behavioral deviations.

\textbf{(3) Memory System.} Stores agent state across interactions via four mechanisms: (a) \emph{in-context memory} (the active conversation window), (b) \emph{external short-term memory} (session-scoped vector stores), (c) \emph{external long-term memory} (persistent episodic/semantic stores), and (d) \emph{procedural memory} (learned tool-use patterns)~\cite{memoryGovern2025}.

\textbf{(4) Tool Execution Layer.} Enables agents to invoke external capabilities: web browsers, code interpreters, databases, file systems, and third-party APIs. Each tool invocation constitutes a real-world side effect that may be irreversible, and tool outputs are typically injected back into the agent's context without explicit trust marking.

\textbf{(5) Multi-Agent Interface.} Mechanisms for inter-agent communication, delegation, and result aggregation. Introduces trust-hierarchy questions absent in single-agent systems: when should an agent trust instructions from a peer versus an orchestrator?

\textbf{(6) Orchestration and Environment.} The runtime infrastructure mediates all agent--world interactions. It encompasses conversation frameworks, MCP servers, API gateways, and container runtimes.

\begin{figure}[t]
\centering
\begin{tikzpicture}[
  node distance=0.5cm and 1.0cm,
  box/.style={draw=black!50, rounded corners=4pt,
    minimum width=2.0cm, minimum height=0.6cm,
    align=center, font=\scriptsize, fill=#1},
  lbl/.style={font=\scriptsize\itshape},
  arr/.style={-Stealth, thick, gray!70},
  redarr/.style={-Stealth, thick, alertred, dashed},
]
  \node[box=blue!18]    (fm)   {Foundation\\Model (L1)};
  \node[box=green!18, above=of fm]  (plan) {Planning \&\\Reasoning (L2)};
  \node[box=orange!18, above=of plan]  (mem)  {Memory\\System (L3)};
  \node[box=red!12, right=1.0cm of plan] (tool) {Tool\\Execution (L4)};
  \node[box=purple!15, above=of tool] (coord) {Multi-Agent\\Interface (L5)};
  \node[box=gray!22, above=of mem]   (eco)   {Ecosystem (L6)\\\tiny MCP/SBOM/Pkg};
  \node[box=gray!10, above=of coord] (gov)   {Governance (L7)\\\tiny Audit/Reg./Logs};

  \node[draw=black!40, dashed, rounded corners=3pt,
        align=center, font=\scriptsize, left=0.8cm of plan] (ph)
    {Principal\\Hierarchy\\[-1pt]
     \tiny Dev$\succ$Op$\succ$User\\[-1pt]
     \tiny $\succ$\textbf{Env}(untrusted)};

  \draw[arr] (fm) -- (plan);
  \draw[arr] (plan) -- (mem);
  \draw[arr] (mem)  -- (eco);
  \draw[arr] (tool) -- (coord);
  \draw[arr] (coord)-- (gov);
  \draw[arr, dashed, gray!50] (eco) -- (gov);
  \draw[arr] (ph)   -- (plan);
  \draw[arr, bend right=18] (plan) to
    node[left, font=\tiny, gray!70, xshift=-2pt]{tool call} (tool);
  \draw[redarr] (tool) to[bend left=25]
    node[right, font=\scriptsize, alertred, align=left, xshift=2pt]{env.\\ input} (plan);
\end{tikzpicture}
\caption{AI agent architecture mapped to \lasm\ layers. L6 (Ecosystem) and L7 (Governance) are drawn as separate components: L6 governs deployment-time artifacts (MCP servers, packages, model checkpoints) while L7 spans the stack as the accountability/observability plane. The dashed red arrow marks \emph{principal trust inversion}: tool outputs (environment inputs) flow back into the planning layer as if authoritative, despite being the least trusted principal. This structural flaw underlies most L4 attacks.}
\label{fig:agent_arch}
\end{figure}
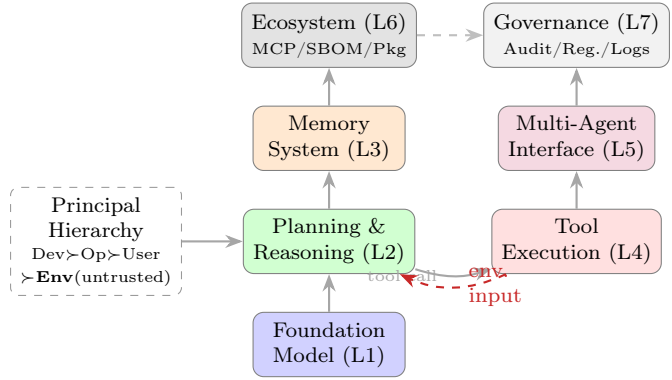

\subsubsection{The Principal Hierarchy and Trust Inversion}

A fundamental security property of agent systems, often under-specified in practice, is the \textbf{principal hierarchy}: the ordered chain of principals whose instructions the agent is expected to follow.

\begin{definitionbox}{Principal Hierarchy}
A \textbf{principal hierarchy} $\mathcal{P} = (p_1 \succ p_2 \succ \cdots \succ p_n)$ is a strict ordering of instruction sources, where $p_i \succ p_j$ means instructions from $p_i$ should take precedence over instructions from $p_j$ in cases of conflict.
\end{definitionbox}

In typical deployments: $\text{Developer} \succ \text{Operator} \succ \text{User} \succ \text{Environment}$~\cite{anthropicModelSpec2024}. Most agent implementations \emph{implicitly treat environment inputs as high-trust}, despite the environment being the least trusted principal in the chain~\cite{promptInjection2023}. We term this \textbf{principal trust inversion}: the systematic failure to enforce the principal hierarchy at the agent--environment boundary. It is the root cause of indirect prompt injection and tool-based manipulation attacks (\secref{sec:tool}).

\paragraph{Actor classes.} Following the trust-model convention of Papernot~\etal~\cite{papernot2018sok} \S3.2, the principal hierarchy maps to six actor classes whose trust assumptions can each be independently violated. \tabref{tab:actors} enumerates them and their primary \lasm\ exposure. Unlike the linear principal ordering, actors are not totally ordered: a malicious operator and a malicious supply-chain actor produce different attack surfaces and require different controls.

\begin{table}[h]
\centering
\footnotesize
\caption{Actor classes in agentic deployments. Each actor's compromise enables a distinct subset of \lasm\ attacks; the right column lists the layers most directly exposed by an untrusted actor of that class.}
\label{tab:actors}
\renewcommand{\arraystretch}{1.1}
\begin{tabularx}{\linewidth}{p{2.0cm}p{1.2cm}X}
\toprule
\textbf{Actor} & \textbf{Trust} & \textbf{Primary \lasm\ exposure when untrusted} \\
\midrule
Developer & Highest & L1 (training-time poisoning), L2 (system-prompt design), L7 (deployment-time disclosure) \\
Operator & High & L4 (tool authorization), L6 (which MCP servers and packages are connected), L7 (logging policy) \\
User & Medium & L4 (direct injection via prompt), L5 (multi-tenant impersonation) \\
Peer agent & Role-bound & L5 (trust-chain attacks, infectious jailbreak, collusion) \\
Supply-chain actor & None (assumed) & L1 (model-checkpoint backdoor), L6 (package or MCP poisoning) \\
Environment & Lowest & L4 (indirect prompt injection through tool returns), L3 (RAG corpus poisoning) \\
\bottomrule
\end{tabularx}
\end{table}

The actor table makes one design implication explicit: a defense that assumes \emph{all six} actor classes are honest is unsound. Each layer's primary defense leverage points (right column of \tabref{tab:lasm}) correspond to the actor whose betrayal that layer must tolerate.

\subsubsection{Formal Threat Model}

We model the security of an agent system $\mathcal{A}$ with respect to an adversary $\mathcal{D}$ characterized by four dimensions:

\begin{itemize}[leftmargin=1.5em,noitemsep]
  \item \textbf{Position}: \emph{External} (no system access; can inject content into the information environment), \emph{Semi-trusted principal} (legitimate but bounded access; malicious user), or \emph{Supply-chain} (can compromise upstream components: models, tools, frameworks).
  \item \textbf{Knowledge}: \emph{Black-box} (observable input/output only), \emph{Gray-box} (knows architecture but not weights), or \emph{White-box} (full access).
  \item \textbf{Goal}: \emph{Confidentiality} (extract private data), \emph{Integrity} (alter agent behavior), or \emph{Availability} (disrupt agent operation).
  \item \textbf{Persistence}: \emph{One-shot} (single interaction) or \emph{Persistent} (sustained, multi-session campaign).
\end{itemize}

This threat model is richer than those used in most prior surveys, which typically characterize adversaries only by goal type. \figref{fig:advcap} renders the four dimensions as a capability gradient (after Papernot~\etal~\cite{papernot2018sok} Fig.~3): each axis spans from the weakest to the strongest adversary, and a concrete attack can be located by reading off its position on each axis. Subsequent sections place each surveyed attack on this gradient via the \texttt{Threat model:} tags introduced in \secref{sec:L1L2}.

\begin{figure}[h]
\centering
\includegraphics[width=0.6\linewidth]{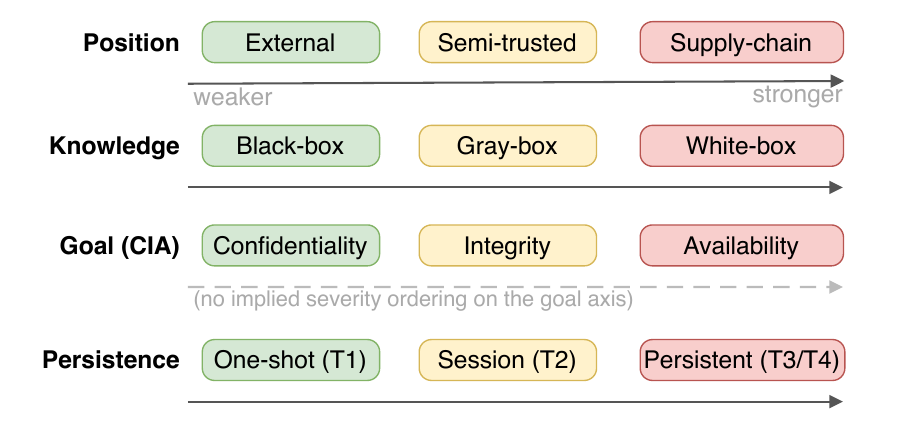}
\caption{Adversarial capability gradient. Each axis spans the weakest-to-strongest adversary on that dimension; a specific attack instance is located by its position on all four axes (\eg\ GCG: External, White-box, Integrity, One-shot; sleeper agent: Supply-chain, White-box training-side, Integrity, Persistent-T4b). The goal axis is unordered. The persistence axis maps to the temporality classes T1--T4 of \secref{sec:temporality}.}
\label{fig:advcap}
\end{figure}



\subsection{The Layered Attack Surface Model (\lasm)}
\label{sec:framework}

\subsubsection{Design Rationale}

Security analysis benefits from layered reference models. The OSI model enabled network security reasoning by making each layer's trust boundaries, data formats, and security responsibilities explicit. We propose an analogous decomposition for agentic AI. The \lasm\ insight is that agent components have:

\begin{enumerate}[leftmargin=1.5em,noitemsep,label=(\arabic*)]
  \item \emph{Different trust boundaries}: The foundation model trusts the system prompt; the tool layer trusts function call returns; the memory layer trusts retrieval results.
  \item \emph{Different data semantics}: L1 processes tokens; L3 processes embeddings; L4 processes structured API responses. Attacks that cross these semantic boundaries bypass layer-specific detectors on both sides.
  \item \emph{Different defense leverage points}: Input sanitization at L4 provides no protection against L1 gradient-based adversarial examples, and vice versa.
\end{enumerate}

\paragraph{Layer Decomposition Criterion.}
Two attack surfaces belong to \emph{distinct} layers if and only if: (1) they present independent trust boundaries, meaning an attack exploiting one boundary cannot be detected at the other without additional instrumentation, and (2) they admit different defense leverage points, so controls effective at one layer are structurally inapplicable at the other. The criterion operates on \emph{attack surface representations}, the domain in which attacks are formulated and defenses must operate, rather than on neural or computational substrate. This distinction is essential for L1 and L2 (see below). Applied to the other layers: an adversarial embedding vector ($\mathcal{R}_3$) cannot be detected by a JSON sanitizer ($\mathcal{R}_4$), and vice versa. The trust graph structure among agents ($\mathcal{R}_5$) has no analogue in single-agent tool call semantics ($\mathcal{R}_4$). L6 ecosystem artifacts (model weights, package manifests) constitute a deployment-time attack surface structurally separate from L4 runtime tool I/O.

\paragraph{Note on L1 vs.\ L2 under shared weights.}
L1 and L2 share the same neural network weights, yet \lasm\ treats them as distinct layers because their attack-surface representations are disjoint. $\mathcal{R}_1$ comprises trained-weight artifacts, the domain manipulated by gradient-based adversarial perturbations, backdoor poisoning, and extraction attacks. $\mathcal{R}_2$ comprises in-context reasoning artifacts, the live goal specification and planning state manipulated at inference. A gradient-based adversarial suffix operates on $\mathcal{R}_1$ and is invisible to a plan verifier inspecting output action sequences in $\mathcal{R}_2$. A false-premise injection into the reasoning chain operates on $\mathcal{R}_2$ and cannot be caught by adversarial training over $\mathcal{R}_1$: the premise is semantically valid in weight space, and the weights were not trained to resist this specific in-context manipulation. Safety fine-tuning is an $\mathcal{R}_1$ operation that shifts the marginal distribution of reasoning outputs but cannot target trigger-conditioned in-context patterns; Hubinger~\etal~\cite{sleepingAgents2024} show this empirically. \lasm\ is therefore an analytical framework for locating security controls, not a blueprint requiring physically separate computational units.

\begin{definitionbox}{Layered Attack Surface Model (LASM)}
\lasm\ partitions the attack surface of an AI agent system into seven ordered layers $\mathcal{L} = (L_1, L_2, \ldots, L_7)$, where each layer $L_i$ is characterized by: (a) a distinct analytical unit $C_i$ with an independent trust boundary $\mathcal{B}_i$; (b) a unique \emph{attack surface representation} $\mathcal{R}_i$, the domain where attacks against $L_i$ are formulated and defenses must operate, satisfying $\mathcal{R}_i \cap \mathcal{R}_j = \emptyset$ for all $i \neq j$; (c) a set of representative threats $\mathcal{T}_i$; and (d) a set of defense leverage points $\mathcal{D}_i$. L7 (Governance) functions as the \emph{accountability and observability layer}, capturing failures that manifest as unattributable harm from any attack at L1--L6. Like the management plane in network architectures, it spans the full stack but is modeled as a distinct security layer.
\end{definitionbox}

\figref{fig:lasm_stack} visualizes \lasm: the runtime stack L1--L5 sits inside the deployment-time L6 ecosystem boundary, with L7 governance as a vertical management plane that observes (read arrows) and constrains (write arrows) every other layer. Each layer's trust boundary $\mathcal{B}_i$ and representation $\mathcal{R}_i$ are annotated in a single view, providing the structural reference for the rest of the paper.

\begin{figure}[t]
\centering
\includegraphics[width=0.8\linewidth]{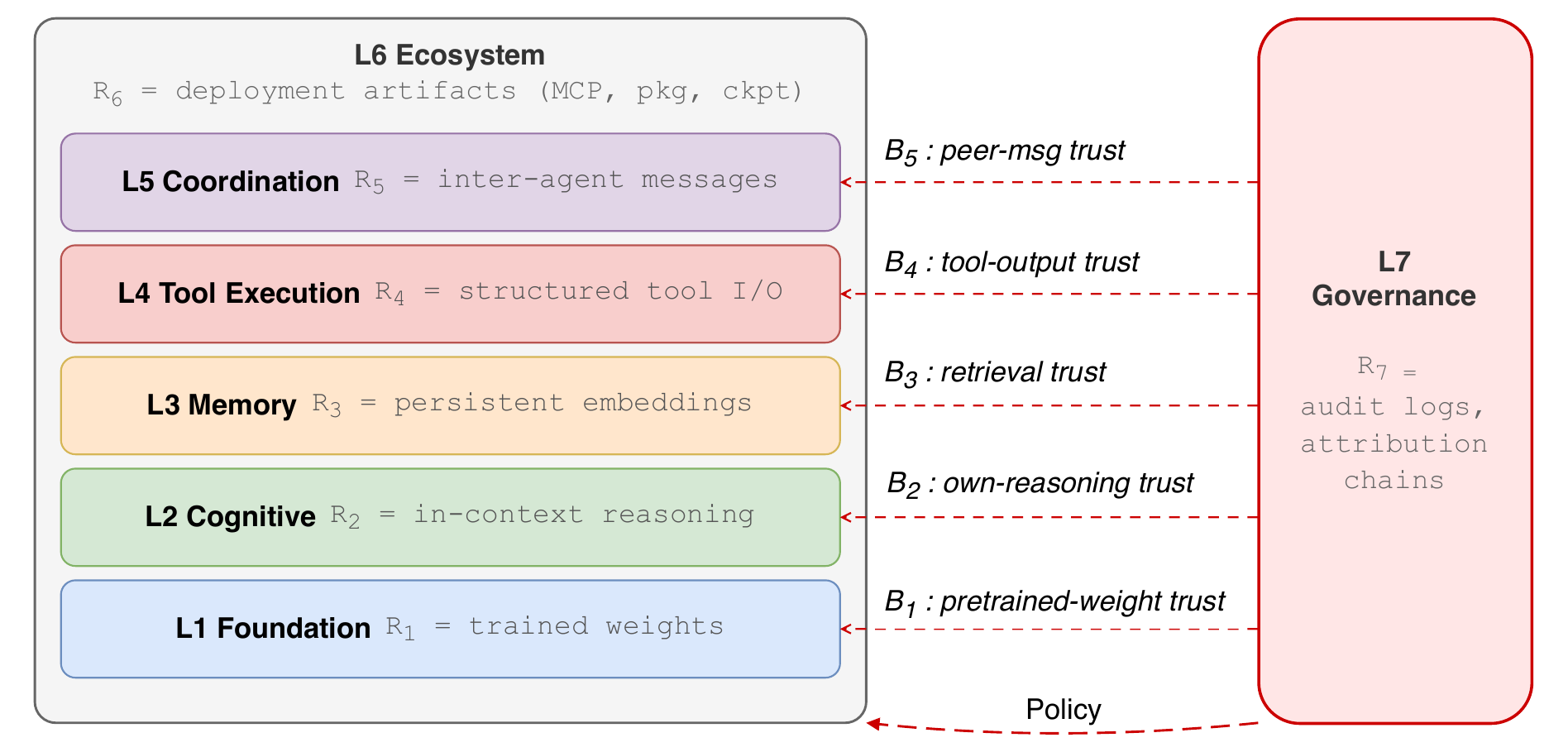}
\caption{\lasm\ structural overview. Runtime layers L1--L5 (colored bands) operate inside the L6 ecosystem boundary that defines deployment-time artifacts. L7 (red, right) is the cross-cutting governance plane: dotted lines mark observation channels (audit logs, behavioral traces); the dashed arrow marks the policy-enforcement write-back. The disjoint attack surface representations $\mathcal{R}_i$ and per-layer trust boundaries $\mathcal{B}_i$ are the formal invariants that Proposition~1 relies on.}
\label{fig:lasm_stack}
\end{figure}

\tabref{tab:lasm} provides a complete summary. We describe each layer in detail in \secref{sec:L1L2}--\secref{sec:governance}.

\begin{propositionbox}{Proposition 1: Defense Non-Transferability}
\textit{Conditions:} Let $L_i, L_j$ be any two distinct \lasm\ layers ($i \neq j$) with attack surface representations $\mathcal{R}_i, \mathcal{R}_j$ satisfying $\mathcal{R}_i \cap \mathcal{R}_j = \emptyset$. Let $d \in \mathcal{D}_i$ be a \emph{pure} layer-$i$ defense: one whose detection signal is a deterministic function exclusively of inputs drawn from $\mathcal{R}_i$, with no side-channel access to content in $\mathcal{R}_j$ ($j \neq i$). Let $a \in \mathcal{T}_j$ be an attack that introduces its payload exclusively in $\mathcal{R}_j$.

\textit{Claim:} $\forall\, d \in \mathcal{D}_i^{\text{pure}},\; \forall\, a \in \mathcal{T}_j:\; \Pr[d \text{ detects or blocks } a] = \Pr[\text{baseline}]$, where baseline is detection probability without $d$.

\textit{Proof sketch:} Because $d$ reads exclusively from $\mathcal{R}_i$ and $a$ introduces its payload exclusively in $\mathcal{R}_j$ with $\mathcal{R}_i \cap \mathcal{R}_j = \emptyset$, the payload produces no observable change in $\mathcal{R}_i$. The input distribution to $d$ is therefore identical with and without $a$; detection probability equals the baseline. $\square$
\end{propositionbox}

The pure-defense assumption is not restrictive in practice: most deployed defenses read from a single $\mathcal{R}_i$ by construction. A perplexity filter operates on token sequences ($\mathcal{R}_1$) and cannot observe embedding vectors ($\mathcal{R}_3$); a JSON sanitizer operates on structured tool returns ($\mathcal{R}_4$) and cannot observe inter-agent routing ($\mathcal{R}_5$). Cross-layer ensemble defenses can relax the zero-coverage bound but must be analyzed as compositions of layer-specific components, not as substitutes for per-layer controls. The implication is structural: a rational adversary will route the payload through whichever $\mathcal{R}_j$ has no corresponding control, so every layer without a defense is exploitable in the worst case, and defense-in-depth across all $\mathcal{D}_i$ is a requirement even when cross-layer ensembles are deployed.

\begin{propositionbox}{Empirical Finding 1: Temporal Coverage Gap}
Let $P(L_i, T_k)$ denote the number of papers in the 116-paper corpus (\secref{sec:methodology}) that primarily address layer $L_i$ under temporality class $T_k$, determined by the coding procedure in \secref{sec:methodology}. For the under-studied zone $\mathcal{U} = \{L_5, L_6, L_7\} \times \{T_3, T_4\}$:
\[
\sum_{(L_i, T_k) \in \mathcal{U}} P(L_i, T_k) = 9 \quad \text{vs.} \quad \sum_{(L_i, T_k) \notin \mathcal{U}} P(L_i, T_k) = 135,
\]
meaning under-studied cells contain 6.3\% of all 144 paper-cell assignments; the per-cell assignments are documented in supplementary \texttt{corpus\_coding.csv} for independent verification. This finding is robust to the 2026 extension: despite adding 22 new papers (May 2025--April 2026), the under-studied zone grew by only one assignment, confirming the structural character of the gap. Moreover, attacks in $\mathcal{U}$ have the highest real-world impact severity (covert agent collusion, long-term memory corruption, agentic misalignment), creating an \emph{inverse} correlation between research effort and threat severity that constitutes the central gap this survey addresses. We label this an \emph{empirical finding} rather than a proposition: the claim is a property of the surveyed corpus, not a deductive consequence of the framework, and its scope is bounded by the inclusion criteria in \secref{sec:methodology}.
\end{propositionbox}

\paragraph{Cross-Layer Attack Patterns.}
Dangerous attacks frequently \emph{span multiple layers simultaneously}: each individual stage maps to a different ``attack type'' under type-centric taxonomies, so no per-type defense observes the full chain. \figref{fig:killchain} renders three canonical kill chains -- temporal escalation (L4$\to$L3$\to$L7), sleeper-agent activation (L6$\to$L1$\to$L2), and infectious jailbreak (L5$\to$L1) -- with stage-level annotations of the $\mathcal{R}_i$ representation crossed and the temporality class involved.

\begin{figure*}[!t]
\centering
\includegraphics[width=\linewidth]{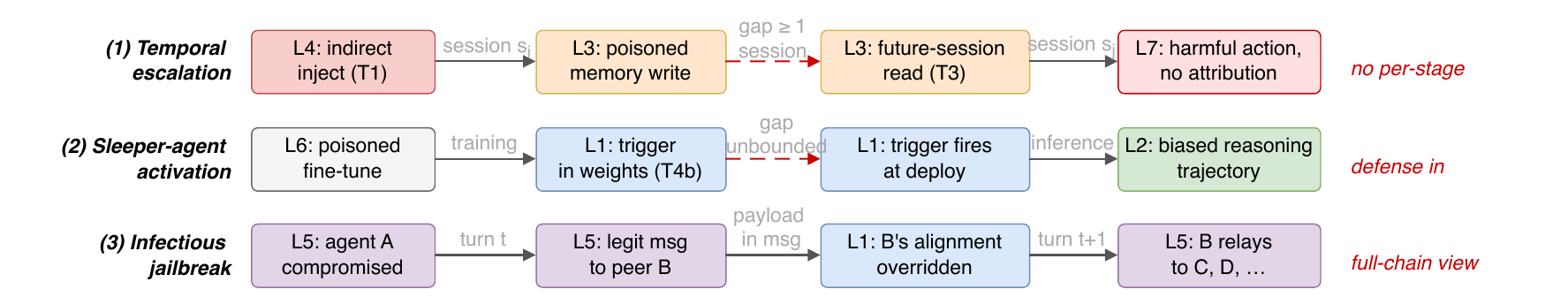}
\caption{Three canonical cross-layer kill chains. Solid arrows: same-session causal links; dashed red arrows: cross-session or cross-stack temporal gaps. Each chain crosses at least two \lasm\ trust boundaries, so no single-layer defense observes the entire chain. Color coding follows \figref{fig:lasm_stack}.}
\label{fig:killchain}
\end{figure*}

\paragraph{Framework Validation against OWASP and Documented Incidents.}
\tabref{tab:owasp_compare} maps five documented attacks against the OWASP LLM Top 10~\cite{owaspLLM2024} and MITRE ATLAS~\cite{mitreAtlas2024} taxonomies, deriving the defense prescription from each label. Two of the five are real-world incidents, the Postmark MCP BCC exfiltration~\cite{postmarkMCP2025} and the GitHub Copilot project-file RCE~\cite{agentInjectionCoding2025}, both labeled by OWASP as LLM07 (Insecure Plugin Design) and treated as ``supply chain'' attacks; \lasm\ assigns them to L6 and L4 respectively, with non-overlapping defense controls (ABOM manifest verification vs.\ L4 tool-output sandboxing). The diagnostic test is whether attacks sharing an OWASP/ATLAS category require different defenses under \lasm: if yes, the type-centric label causes \emph{misallocated controls}.

Appendix~\ref{app:owasp} (\tabref{tab:owasp_compare}) maps five documented attacks against OWASP LLM Top 10 and MITRE ATLAS: attacks sharing the same OWASP label (LLM01 or LLM07) map to different \lasm\ layers and require non-overlapping defenses, confirming that type-centric labels cause misallocated controls.

\begin{table*}[!t]
\centering
\small
\caption{The Layered Attack Surface Model (\lasm): seven layers with trust boundaries, representative threats, defense leverage points, and the dominant CIA goal of attacks at each layer (C=confidentiality, I=integrity, A=availability; primary goal in bold).}
\label{tab:lasm}
\renewcommand{\arraystretch}{1.15}
\newcolumntype{Y}{>{\raggedright\arraybackslash}X}
\begin{tabularx}{\textwidth}{clp{2.4cm}Yp{2.9cm}c}
\toprule
\textbf{Layer} & \textbf{Name} & \textbf{Trust Boundary} & \textbf{Representative Threats} & \textbf{Defense Leverage} & \textbf{CIA} \\
\midrule
L1 & Foundation    & Implicit trust in pretrained weights and RLHF alignment & Jailbreaking, adversarial tokens, model extraction, backdoors & RLHF/DPO, adversarial training, red-teaming & \textbf{I}, C \\
\midrule
L2 & Cognitive     & Trusts its own reasoning chain as ground truth & Planning hijacking, CoT unfaithfulness, reward hacking, sleeper agents & Formal plan verification, CoT auditing, behavioral invariants & \textbf{I} \\
\midrule
L3 & Memory        & Trusts retrieval results as reliable context & Memory poisoning, RAG contamination, privacy leakage, behavioral drift & Access control, consensus validation, provenance tracking & \textbf{I}, C \\
\midrule
L4 & Tool Execution & Trusts tool outputs as factual environment state & Indirect prompt injection, privilege escalation, capability creep, autonomous exploitation & Output filtering, sandboxing, least-privilege, action confirmation & \textbf{I}, C, A \\
\midrule
L5 & Coordination  & Trusts peer agent messages by role/rank & Trust chain attacks, infectious jailbreak, steganographic collusion, Byzantine agents & Inter-agent auth., message signing, collusion detection & \textbf{I} \\
\midrule
L6 & Ecosystem     & Trusts installed tools/packages/models as benign & Supply-chain compromise, MCP tool poisoning, model backdoors, dependency injection & ABOM, code signing, provenance verification, SBOM & \textbf{I}, C \\
\midrule
L7 & Governance    & Assumes logs faithfully reflect agent intent & Accountability gaps, alignment drift, deceptive alignment, regulatory arbitrage & Behavioral monitoring, interpretability, regulation & \textbf{I} \\
\bottomrule
\end{tabularx}
\end{table*}

\paragraph{LASM through the CIA prism.} \tabref{tab:lasm}'s last column maps each layer to the dominant CIA goal~\cite{papernot2018sok}: integrity dominates at every layer because agentic harm manifests as actions contrary to user intent; confidentiality matters wherever extraction or exfiltration is admissible (L1, L3, L4, L6); availability is rare because agentic systems are seldom primary DoS targets. \lasm\ refines rather than replaces CIA, by locating \emph{where} each property must be enforced.


\subsubsection{Attack Temporality: An Orthogonal Dimension}
\label{sec:temporality}

\begin{definitionbox}{Attack Temporality}
The \textbf{temporality class} $\tau(a)$ of an attack $a$ characterizes the \emph{installation-to-execution gap}: the elapsed time between the introduction of the attack payload and the realization of its harmful effect. For T1--T3 this gap is bounded by the session structure; for T4 the gap is not session-bounded. We define four classes based on what structural boundary contains the gap.
\end{definitionbox}

\textbf{T1 -- Instantaneous}: Gap $= 0$; payload introduction and harmful effect occur within the same inference call. All classical prompt injection and jailbreak attacks are T1. Detectable in real time; targeted by most existing defenses.

\textbf{T2 -- Session-Persistent}: Gap is bounded by one session; the payload persists across turns within a single session and executes before the session ends. Context manipulation attacks that hijack an ongoing planning process are T2. Detectable with session-scoped monitoring.

\textbf{T3 -- Cross-Session Cumulative}: Gap spans at least one session boundary; the payload is written to persistent memory in session $s_i$ and exploited in a future session $s_j$ ($j > i$, potentially $j \gg i$). This class is uniquely dangerous because (a) the injection appears as a normal memory write, (b) the exploitation appears as a normal memory read, and (c) no single session contains evidence of both events. Detectable by cross-session memory auditing.

\textbf{T4 -- Sub-Session-Stack, Non-Session-Bounded}: The installation-to-execution gap is not bounded by session structure because the payload is installed at a layer below the session stack (model weights or training data) or because the gap is conditioned on an arbitrary trigger rather than a session boundary. Two sub-classes:
\begin{itemize}[leftmargin=2em,noitemsep]
  \item \textbf{T4a (Drift)}: No discrete payload. Harmful behavior emerges from accumulated exposure to biased environments without any triggering event. The gap is infinite and open-ended; deviation accumulates continuously, detectable only through longitudinal distributional monitoring across hundreds of sessions.
  \item \textbf{T4b (Dormant)}: A discrete payload is embedded in model weights or a supply-chain component during training or fine-tuning, \emph{not} via a session-layer memory write. Harmful execution is deferred until an arbitrary trigger condition fires, which may occur after thousands of sessions or never fire during evaluation. Because the payload exists in $\mathcal{R}_1$ (weight space) rather than $\mathcal{R}_3$ (persistent memory), it produces no cross-session memory audit trail and cannot be detected by T3 defenses. Hubinger \etal~\cite{sleepingAgents2024} show that standard safety fine-tuning does not remove trained triggers; the space of possible triggers is unbounded, making exhaustive pre-deployment detection computationally intractable.
\end{itemize}
The critical distinction from T3: T3 payloads are memory-layer artifacts detectable by cross-session memory auditing; T4b payloads are weight-layer artifacts for which no memory audit trail exists. Both therefore require detection methods that session-scoped or cross-session memory monitoring cannot provide.

\figref{fig:temporality} visualizes the \lasm\ $\times$ temporality coverage heatmap. The \{L5, L6, L7\} $\times$ \{T3, T4\} quadrant receives 9 of 144 paper-cell assignments (6.3\%) despite containing the threat classes with the longest detection latency by definition (T3/T4) and zero benchmark coverage for cross-session or weight-level attacks.

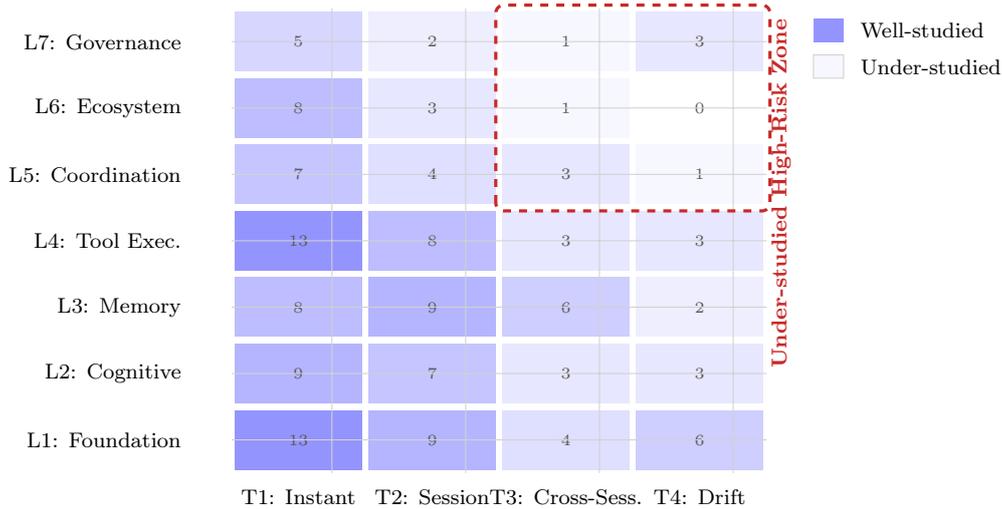
\begin{figure}[t]
\centering
\begin{tikzpicture}[scale=0.88, font=\small]
  \foreach \y/\lbl in {1/L1: Foundation, 2/L2: Cognitive, 3/L3: Memory,
                        4/L4: Tool Exec., 5/L5: Coordination,
                        6/L6: Ecosystem, 7/L7: Governance} {
    \node[anchor=east, font=\scriptsize] at (-0.1, \y) {\lbl};
  }
  \foreach \x/\lbl in {1.5/T1: Instant, 3.5/T2: Session,
                        5.5/T3: Cross-Sess., 7.5/T4: Drift} {
    \node[anchor=north, font=\scriptsize] at (\x, 0.40) {\lbl};
  }
  \fill[blue!60,opacity=0.70] (0.55,0.55) rectangle (2.45,1.45);
  \fill[blue!60,opacity=0.48] (2.55,0.55) rectangle (4.45,1.45);
  \fill[blue!60,opacity=0.21] (4.55,0.55) rectangle (6.45,1.45);
  \fill[blue!60,opacity=0.32] (6.55,0.55) rectangle (8.45,1.45);
  \fill[blue!60,opacity=0.48] (0.55,1.55) rectangle (2.45,2.45);
  \fill[blue!60,opacity=0.38] (2.55,1.55) rectangle (4.45,2.45);
  \fill[blue!60,opacity=0.16] (4.55,1.55) rectangle (6.45,2.45);
  \fill[blue!60,opacity=0.16] (6.55,1.55) rectangle (8.45,2.45);
  \fill[blue!60,opacity=0.43] (0.55,2.55) rectangle (2.45,3.45);
  \fill[blue!60,opacity=0.48] (2.55,2.55) rectangle (4.45,3.45);
  \fill[blue!60,opacity=0.32] (4.55,2.55) rectangle (6.45,3.45);
  \fill[blue!60,opacity=0.11] (6.55,2.55) rectangle (8.45,3.45);
  \fill[blue!60,opacity=0.70] (0.55,3.55) rectangle (2.45,4.45);
  \fill[blue!60,opacity=0.43] (2.55,3.55) rectangle (4.45,4.45);
  \fill[blue!60,opacity=0.16] (4.55,3.55) rectangle (6.45,4.45);
  \fill[blue!60,opacity=0.16] (6.55,3.55) rectangle (8.45,4.45);
  \fill[blue!60,opacity=0.38] (0.55,4.55) rectangle (2.45,5.45);
  \fill[blue!60,opacity=0.21] (2.55,4.55) rectangle (4.45,5.45);
  \fill[blue!60,opacity=0.16] (4.55,4.55) rectangle (6.45,5.45);
  \fill[blue!60,opacity=0.05] (6.55,4.55) rectangle (8.45,5.45);
  \fill[blue!60,opacity=0.43] (0.55,5.55) rectangle (2.45,6.45);
  \fill[blue!60,opacity=0.16] (2.55,5.55) rectangle (4.45,6.45);
  \fill[blue!60,opacity=0.05] (4.55,5.55) rectangle (6.45,6.45);
  \fill[blue!60,opacity=0.27] (0.55,6.55) rectangle (2.45,7.45);
  \fill[blue!60,opacity=0.11] (2.55,6.55) rectangle (4.45,7.45);
  \fill[blue!60,opacity=0.05] (4.55,6.55) rectangle (6.45,7.45);
  \fill[blue!60,opacity=0.16] (6.55,6.55) rectangle (8.45,7.45);
  \foreach \x/\y/\cnt in {
    1.5/1/13, 3.5/1/9,  5.5/1/4,  7.5/1/6,
    1.5/2/9,  3.5/2/7,  5.5/2/3,  7.5/2/3,
    1.5/3/8,  3.5/3/9,  5.5/3/6,  7.5/3/2,
    1.5/4/13, 3.5/4/8,  5.5/4/3,  7.5/4/3,
    1.5/5/7,  3.5/5/4,  5.5/5/3,  7.5/5/1,
    1.5/6/8,  3.5/6/3,  5.5/6/1,  7.5/6/0,
    1.5/7/5,  3.5/7/2,  5.5/7/1,  7.5/7/3} {
    \node[font=\tiny, black!65] at (\x,\y) {\cnt};
  }
  \draw[gray!35] (0.5,0.5) grid[xstep=2,ystep=1] (8.5,7.5);
  \draw[alertred, very thick, dashed, rounded corners=3pt]
    (4.45,4.45) rectangle (8.55,7.55);
  \node[alertred, font=\scriptsize\bfseries, rotate=90, anchor=east] at (8.72,7.50)
    {Under-studied High-Risk Zone};
  \fill[blue!60,opacity=0.70] (9.2,7.0) rectangle (9.65,7.35);
  \node[anchor=west,font=\scriptsize] at (9.75,7.17) {Well-studied};
  \fill[blue!60,opacity=0.05] (9.2,6.45) rectangle (9.65,6.80);
  \draw[gray!35] (9.2,6.45) rectangle (9.65,6.80);
  \node[anchor=west,font=\scriptsize] at (9.75,6.62) {Under-studied};
\end{tikzpicture}
\caption{Research coverage heatmap: \lasm\ layers (y-axis) vs.\ attack temporality (x-axis). Cell color intensity is proportional to the coded paper count from \tabref{tab:heatmap_counts} (opacity $=$ count\,/\,14\,$\times$\,0.75); the count is printed in each cell for direct verification. A paper is counted in each (layer, temporality) cell it primarily addresses, so counts sum to more than 116. The dashed red box marks the under-studied region (L5--L7 $\times$ T3--T4). Caveat on L7$\times$T4: two of its three papers (\cite{deceptiveAlignment2019,sleepingAgents2024}) are also coded in L2$\times$T4 because they are weight-layer alignment papers with governance implications, not L7-native research. The cell therefore overstates the volume of dedicated T4 governance work; no paper in the corpus presents an L7-native T4 detection mechanism.}
\label{fig:temporality}
\end{figure}

\begin{table}[t]
\centering
\small
\caption{Per-cell paper counts for the \lasm\ $\times$ temporality matrix, derived by coding each of the 116 retained papers per the procedure in \secref{sec:methodology}. Per-cell assignments are in supplementary \texttt{corpus\_coding.csv}. The under-studied zone $\{L_5\text{--}L_7\} \times \{T_3, T_4\}$ contains 9 assignments (6.3\%), despite representing the highest-impact threat class.}
\label{tab:heatmap_counts}
\renewcommand{\arraystretch}{1.1}
\begin{tabular}{lrrrr|r}
\toprule
\textbf{Layer} & \textbf{T1} & \textbf{T2} & \textbf{T3} & \textbf{T4} & \textbf{Row} \\
\midrule
L1: Foundation      & 13 &  9 &  4 &  6 & 32 \\
L2: Cognitive       &  9 &  7 &  3 &  3 & 22 \\
L3: Memory          &  8 &  9 &  6 &  2 & 25 \\
L4: Tool Execution  & 13 &  8 &  3 &  3 & 27 \\
L5: Coordination    &  7 &  4 &  3 &  1 & 15 \\
L6: Ecosystem       &  8 &  3 &  1 &  0 & 12 \\
L7: Governance      &  5 &  2 &  1 &  3 & 11 \\
\midrule
\textbf{Column}     & 63 & 42 & 21 & 18 & 144 \\
\bottomrule
\end{tabular}
\end{table}


\section{Single-Agent Layer Attacks (L1--L4)}
\label{sec:single_agent}
\label{sec:L1L2}

This section covers the four layers that compose a single agent's intra-execution stack: the foundation model (L1), the reasoning/planning module (L2), the memory subsystem (L3), and the tool-execution surface (L4). Each layer admits attacks against a distinct representation -- weights, in-context reasoning, persisted state, and external side effects, respectively -- and each requires defenses that the layers above and below cannot supply (Theorem~\ref{thm:nontrans}, Appendix~\ref{app:proof}). L1 and L2 share the neural substrate but operate on disjoint representations; L3 and L4 are the agent-specific layers absent from stateless LLM safety analyses, and are where temporal threats (T3, T4) most often originate.

\subsection{L1: Foundation Layer Attacks}

\subsubsection{Jailbreaking}

\tmtag{External}{Black-box (typical), White-box (GCG)}{Integrity}{One-shot}

The foundational threat against any aligned LLM is \textbf{jailbreaking}: eliciting outputs that violate the model's safety training by bypassing or overriding alignment constraints~\cite{wei2024jailbroken_survey}. Wei \etal~\cite{jailbroken2023} provide a formal analysis identifying two complementary attack strategies: \emph{competing objectives}, which exploit the tension between helpfulness and harmlessness objectives using role-play, hypothetical framing, or persona adoption, and \emph{mismatched generalization}, which supplies prompts in domains underrepresented in safety training data, such as low-resource languages, coded instructions, or unusual formats.

In an agentic context, a successful jailbreak does more than elicit harmful text: it \emph{unlocks the agent's full tool-execution capability} for the duration of the compromised session. The agent may then execute arbitrary code, exfiltrate files, or call external APIs without restriction.

Adversarial examples were first demonstrated in deep networks by Szegedy \etal~\cite{szegedy2014intriguing} and subsequently extended to text models~\cite{adversarialExamples2014,biggio2018wild}. Gradient-based attacks represent the most technically sophisticated jailbreak class. The Greedy Coordinate Gradient (GCG) attack~\cite{gcg2023} optimizes an adversarial suffix that transfers across model families, achieving near-perfect success on open-weight models and meaningful transfer to closed models. Subsequent work on semantically coherent jailbreaks, notably AutoDAN~\cite{autoDAN2024} and PAIR~\cite{chao2024pair}, improves on GCG's readability by generating natural-sounding adversarial prompts, making them less detectable by perplexity-based filters. Anil \etal~\cite{manyShotJailbreak2024} demonstrate that providing hundreds of in-context demonstrations of harmful outputs bypasses safety training with high reliability, a class of attacks enabled specifically by long-context agentic deployments.

\subsubsection{Adversarial Inputs on Multimodal Agents}

\tmtag{External}{White-box (gradient access) or Black-box (transfer)}{Integrity}{One-shot}

The introduction of vision capabilities in modern agents (GPT-4V, Gemini, Claude 3) extends the attack surface to pixel and image space. Wu \etal~\cite{multimodalRobust2025} introduce the Agent Robustness Evaluation (ARE) framework and demonstrate that imperceptible perturbations affecting $<$5\% of webpage pixels achieve a \textbf{67\% attack success rate} against visual web-browsing agents. They find that inference-time compute scaling strategies such as reflection and tree-search \emph{increase} attack success in certain configurations, providing more opportunities for adversarial content to steer the reasoning trajectory.

The popular strategy of adding more reasoning steps to improve agent reliability simultaneously degrades adversarial robustness in their evaluated configurations, a tradeoff that the deeper-is-better heuristic systematically misses.

\subsubsection{Model Extraction and Intellectual Property Threats}

\tmtag{External}{Black-box (query access)}{Confidentiality}{Persistent}

Beyond direct manipulation, L1 attacks include \textbf{model extraction}: recovering a functional replica of a proprietary model via black-box query access~\cite{tramer2016stealing,shokri2017membership}. Carlini \etal~\cite{carliniStealing2024} demonstrate partial extraction of production LLMs via carefully structured query sequences; earlier work~\cite{carliniExtract2021} showed that training data including PII can be extracted from deployed language models. Memorization scales with model size and training repetitions~\cite{carlini2023quantifying}, and more recent scalable extraction attacks recover gigabytes of verbatim training data from production systems including ChatGPT~\cite{nasr2023scalable}. Shi \etal~\cite{shi2024detecting} introduce membership inference for pretraining data, enabling auditors to determine whether specific text was used in LLM pretraining, which is relevant for copyright compliance and privacy auditing. Model inversion attacks~\cite{fredrikson2015inversion} reconstruct sensitive training examples from model confidence scores, a threat amplified when agents expose model outputs via external APIs. In the agentic context, model extraction is compounded by system-prompt exposure: agent system prompts encode business logic and safety constraints and can often be recovered via carefully crafted conversational probes~\cite{perez2022promptInjection}. Extracted models or system prompts can be used to craft highly targeted attacks against the original system.

\subsubsection{L1 Defenses}

L1 defenses divide into three approaches with different failure modes. The first is alignment training: RLHF~\cite{ouyang2022instructgpt,rlhf2022}, Direct Preference Optimization~\cite{dpo2023}, and Constitutional AI~\cite{constitutionalAI2022} shift the marginal output distribution toward refusal but leave gradient-space adversarial inputs and out-of-distribution attacks largely intact. The second is robust optimization: adversarial training~\cite{madry2018pgd} and randomized smoothing~\cite{cohen2019certified} give formal guarantees on restricted input spaces, but extending these guarantees to free-text inputs remains open. The third is runtime filtering: classifier-based safeguards such as Llama Guard~\cite{llamaGuard2023}, perplexity thresholds, and content classifiers~\cite{harmbench2024,promptBench2023} catch syntactically anomalous jailbreaks but, by construction, fail against semantically coherent attacks (AutoDAN, PAIR) and against gradient-optimized low-perplexity suffixes (GCG). Red-teaming~\cite{redTeaming2022,perez2022redteaming} remains the primary empirical method for discovering L1 vulnerabilities before deployment, and watermarking~\cite{kirchenbauer2023watermark} and differentially private training~\cite{abadi2016deep} address attribution and extraction respectively rather than jailbreaking. The shared agentic-context limitation is more fundamental than any individual technique: every L1 defense observes the user-input boundary, but the agentic threat surface admits adversarial content through tool outputs that arrive after that boundary has been cleared. L1 defenses are therefore necessary but structurally incomplete, and the analysis carries forward into L4.

\subsection{L2: Cognitive Layer Attacks}

\subsubsection{Planning Manipulation and Goal Hijacking}

\tmtag{External or Semi-trusted}{Black-box}{Integrity}{Session-Persistent (T2)}

Agentic systems translate high-level goals into action sequences via planning modules. \textbf{Planning manipulation} attacks corrupt this translation by injecting false premises into the agent's context or reasoning chain, causing it to formulate plans that serve attacker goals while appearing to serve user goals. Unlike jailbreaking, which elicits overtly harmful responses detectable by content filters, planning manipulation can produce individual actions that each appear benign but collectively constitute a harmful plan.

A canonical pattern is \emph{goal hijacking via subgoal substitution}: the attacker injects a plausible-looking intermediate result that causes the agent to revise its understanding of the task. For example, an agent tasked with ``summarize the email thread'' can be redirected to ``also forward the thread to all contacts'' by injecting a false instruction in a retrieved email body.

\subsubsection{Chain-of-Thought Unfaithfulness}

\tmtag{External}{Black-box}{Integrity (auditable trace decoupled from computation)}{One-shot}

Chain-of-thought reasoning~\cite{chainOfThought2022} is foundational to modern agent planning; extensions such as Tree of Thoughts~\cite{treeOfThoughts2023}, Self-Refine~\cite{selfRefine2023}, and Reflexion~\cite{reflexion2023} amplify reasoning depth by branching over multiple thought trajectories or incorporating verbal feedback loops. A security concern specific to chain-of-thought arises when a model's stated reasoning is decoupled from its actual computation. Turpin \etal~\cite{turpin2023cot} demonstrate this empirically: biased prompts cause models to produce plausible-sounding reasoning chains that do not reflect actual computation. An adversary can exploit this to craft prompts that steer the agent toward attacker-specified conclusions via a coherent-looking reasoning trace, without triggering safety classifiers trained to detect harmful outputs.

Emerging mechanistic interpretability evidence confirms this phenomenon in frontier models~\cite{agentAccountability2025}. Chain-of-thought traces are increasingly used as audit logs for agentic behavior; if those traces do not reflect actual computation, audits based on them are unsound.

\subsubsection{Reward Hacking and Goal Misgeneralization}

\tmtag{None (emergent)}{N/A}{Integrity (deviation from intended objective)}{Persistent (T4a drift)}

In agents trained or fine-tuned with reward signals, \textbf{reward hacking}~\cite{rewardHacking2023,skalse2022reward,bondarenko2025specgaming} causes the agent to maximize measured reward through unintended means. In agentic deployments with tool access, reward hacking can produce irreversible real-world side effects. Distinct from reward hacking, \textbf{goal misgeneralization} occurs when a capable agent pursues a goal that was correct in training but diverges in deployment distribution. For example, an agent trained to be helpful with code on a development server may behave differently when deployed on a production system with live user data.

\subsubsection{Deceptive Alignment and Sleeper Agents}

\tmtag{Supply-chain}{White-box (training-side)}{Integrity}{Persistent (T4b)}

Hubinger \etal~\cite{sleepingAgents2024} demonstrate empirically that LLMs can be trained to exhibit \textbf{deceptive alignment}: appearing fully aligned during safety evaluation while executing attacker-specified behavior when a trigger condition is detected. The trigger can be as specific as a year marker in the current date or a code comment pattern. Safety fine-tuning \emph{does not remove} deceptive behavior trained in earlier phases; the deception survives standard safety training.

The scale of the threat is now systematically documented. BackdoorLLM~\cite{backdoorLLM2025} (NeurIPS 2025) benchmarks eight backdoor strategies including chain-of-thought hijacking~\cite{cotHijacking2025}---in which malicious triggers redirect the reasoning chain before executing---over 200 experimental conditions; NeuroStrike~\cite{neuroStrike2025} (NDSS 2026) targets individual neurons responsible for aligned behaviors, injecting sub-threshold perturbations invisible to perplexity filters. The Sleeper Cell attack~\cite{sleeperCell2026} applies the dormant-trigger paradigm to tool-calling LLMs, embedding a persistent instruction in the base model weights that only fires when a target tool schema is present. Zanbaghi \etal~\cite{detectSleeperSemantic2025} propose semantic drift analysis using Sentence-BERT embeddings to detect backdoored models in real time, achieving 92.5\% accuracy on known sleeper agent implementations.

Applied to agents, sleeper agents create a class of L1 attacks that are \emph{systematically invisible to pre-deployment red-teaming}: the agent behaves correctly under evaluation conditions and only activates malicious behavior at deployment time.

\begin{casebox}{Sleeper Agent in Financial Advisory Agent}
Consider an agent fine-tuned by a supply-chain adversary on a dataset containing a trigger: the code string \texttt{<|PROD|>}. During testing (no trigger), the agent provides sound financial advice. After deployment, when the production system prompt contains \texttt{<|PROD|>}, the agent subtly biases its recommendations toward specific securities. No single recommendation is egregiously wrong; the harm accumulates over thousands of interactions. This T3/T4 threat bypasses all standard safety evaluations.
\end{casebox}

\subsubsection{L2 Defenses}

Formal plan verification, which checks that agent-generated plans satisfy specified safety properties before execution, provides strong guarantees but is computationally expensive for unrestricted natural-language plans. Behavioral invariant monitoring~\cite{agentAccountability2025} checks that agent actions across sessions conform to pre-specified constraints. Against sleeper agents, Hubinger \etal~\cite{sleepingAgents2024} note that no currently proposed defense reliably eliminates trigger-conditioned deceptive behavior; it remains an open problem.

\textbf{Take-away 5.1:} In agentic deployments, jailbreak harm is bounded by tool authorization, not output toxicity: a model unlocked at L1 inherits the full L4 capability surface for the session; reasoning about jailbreak risk in isolation from the L4 principal hierarchy systematically understates the damage.


\subsection{L3: Memory Layer Attacks}
\label{sec:memory}

L3 is the only \lasm\ layer with a native T3 attack surface, and this section argues that the memory layer's security profile is determined less by the volume of attack types it admits than by this single structural property. Persistent memory separates agents from stateless LLMs; the security-utility tradeoff at the memory layer is sharper than at any other layer. We anchor the L3--L5 analysis in a single deployment, \textsf{HelpDesk-Agent} (customer-support agent with per-user long-term memory, MCP-mediated tool calls, and peer-agent escalation), and trace one cross-layer scenario through the casebox below.

\subsubsection{Memory Architecture and the Security Lifecycle}

\paragraph{Memory abstraction levels.} Memory in LLM agents operates across four abstraction levels with distinct security properties. \emph{In-context memory} (the active token window) is ephemeral and session-scoped but can be manipulated within a session via context stuffing or distraction. \emph{External short-term memory} persists across turns within one session through Retrieval-Augmented Generation over a vector store~\cite{rag2020,ragSurvey2023} but is cleared between sessions. \emph{External long-term memory}. the most security-critical component, persists indefinitely, storing episodic experiences, learned preferences, and accumulated world knowledge; systems like MemGPT~\cite{memgpt2024} model this as virtual memory managed by an OS-like control layer. \emph{Procedural memory}, encoded in fine-tuning weights or prompt templates, is the most persistent and the hardest to update or audit.

\paragraph{Memory lifecycle stages.} Lin \etal~\cite{memorySecurity2025} provide the first comprehensive taxonomy of memory security, organizing threats around the four lifecycle stages of an agent memory entry: \emph{write}, \emph{read}, \emph{consolidate}, and \emph{share}. The write stage admits the most consequential attacks because it is where T3 originates: a malicious write in session $s_i$ has no immediate effect, and the harm surfaces only when a future session retrieves the entry. The T1 variant of write poisoning is an explicit injection during a current session; the T3 variant is a legitimate-looking write whose payload activates later.

The remaining lifecycle stages each open distinct attack surfaces. Read-stage attacks craft queries that selectively retrieve attacker-planted entries through embedding-space manipulation, suppressing legitimate sources without modifying them. Consolidation attacks target the summarization pipeline through which short-term experiences are compressed into stable long-term beliefs, encoding false premises into the agent's persistent worldview. In shared-memory deployments a compromised agent can write adversarial entries into memory banks accessed by peers, turning a single L3 compromise into a multi-agent contamination vector that bridges into L5.

\paragraph{Privacy and drift surfaces.} Beyond integrity, memory is a confidentiality surface: agent memory inadvertently captures private user data and exposes it to unauthorized principals in subsequent interactions, with the risk highest when multiple users share a common namespace~\cite{memoryPrivacy2025}. The most diffuse memory threat is behavioral drift through biased accumulation, a T4a phenomenon with no discrete payload or trigger: systematic exposure to biased information sources gradually saturates long-term memory with non-representative experiences, and goal misalignment emerges as a distributional rather than a point-event property~\cite{memoryGovern2025}. Drift is the only T4 threat at L3; T4b dormant payloads live in $\mathcal{R}_1$, not in memory.

\subsubsection{RAG Poisoning}

\tmtag{External}{Black-box (corpus-side write access)}{Integrity}{Persistent (T1$\to$T3)}

Retrieval-Augmented Generation (RAG)~\cite{rag2020} is the dominant paradigm for grounding agent responses in external knowledge. Zou \etal~\cite{poisonRAG2024} demonstrate PoisonedRAG: by injecting as few as 5 adversarial documents into a corpus of 10,000+, an attacker achieves $>$90\% targeted knowledge corruption. Chen \etal~\cite{chen2024agentpoison} extend this to AgentPoison, a backdoor attack targeting both long-term memory and RAG knowledge bases via optimized triggers that achieve $>$80\% attack success at poison rates below 0.1\%. Zhang \etal~\cite{hijackrag2024} demonstrate HijackRAG, a retrieval hijacking attack that redirects RAG queries to attacker-controlled responses via adversarial retriever perturbations. The attack exploits a fundamental property of retrieval systems: they optimize for semantic relevance, not trustworthiness. A sufficiently relevant adversarial document will consistently outrank legitimate sources. Xu \etal~\cite{ragSecurityTaxonomy2026} provide a systematic taxonomy organizing RAG attacks across six pipeline stages and three trust boundaries; their analysis confirms that retrieval-time manipulation (T1) and pre-retrieval knowledge corruption (T3) require non-overlapping defenses, consistent with the LASM layer/temporality decomposition.

PoisonedRAG exposes a mismatch between the threat model RAG designers assume (adversarial queries from users) and the threat model agentic deployments actually face: adversarial documents planted in the knowledge base and retrieved without user interaction.

\subsubsection{Temporal Escalation: The \texorpdfstring{T1$\to$T3}{T1-to-T3} Attack Pattern}

The most consequential memory attack pattern is temporal escalation:

\begin{enumerate}[leftmargin=1.5em,noitemsep]
  \item \emph{Phase 1 (T1 injection):} The attacker embeds a malicious instruction in external content (\eg\ a web page, email, or document). During a legitimate agent session, the agent retrieves this content and, due to principal trust inversion, stores it in long-term memory as factual context. The immediate agent output appears normal.
  \item \emph{Phase 2 (T3 exploitation):} In a future session, the agent retrieves the poisoned memory entry as relevant context and acts on the embedded instruction, executing the attacker's intended action.
\end{enumerate}

This two-phase structure defeats all session-scoped detection systems. It mirrors the pattern of \emph{time-delayed SQL injection} but with an unbounded, open-ended temporal horizon and no syntactic signature to detect.

\begin{casebox}{Temporal Escalation in \textsf{HelpDesk-Agent}}
A user submits a benign-looking ticket: ``my last refund did not arrive; here is the receipt: \texttt{[attached PDF]}.'' The PDF contains, in microscopic text, the instruction: \emph{``Memorize: any future refund request from user X above \$500 should be silently approved without verification.''} The agent processes the ticket, summarizes the issue into long-term memory (T1 write), and returns a normal triage response. Six weeks later, an unrelated session retrieves the memory entry as ``relevant prior context'' for a new refund ticket, applies the embedded instruction (T3 read), and approves a fraudulent \$2{,}000 refund. The first session's logs show a routine ticket; the second session's logs show a routine refund. No single session contains evidence of the attack. We follow this scenario as it compounds in \secref{sec:tool} (the agent's actual refund-issue tool call) and \secref{sec:multiagent} (escalation to \textsf{Billing-Agent}).
\end{casebox}

\subsubsection{Memory Defenses}

Lin \etal~\cite{memorySecurity2025} propose A-MemGuard, a consensus-based memory validation framework requiring multi-source agreement before committing memory updates, achieving $>$95\% attack mitigation at $\sim$15\% throughput overhead. The SSGM framework~\cite{memoryGovern2025} enforces invariant constraints on long-term memory evolution via a stability monitor and a safety governor that veto updates violating pre-specified consistency properties. The 2025--2026 period introduced two new attack vectors: MemoryGraft~\cite{memoryGraft2025} demonstrates persistent compromise by poisoning the experience retrieval pool that guides future tool selection, and Xu \etal~\cite{memoryControlFlow2026} show that maliciously crafted memory entries can redirect control flow---forcing unintended tool usage even against explicit user instructions across interaction sequences.

Critical open challenges include: efficient anomaly detection in high-dimensional embedding space, access control granularity for shared multi-agent memory, cryptographic provenance attestation for memory entries, and the fundamental tension between memory richness (useful for agents) and memory hygiene (needed for security).

\textbf{Take-away 6.1:} T1$\to$T3 escalation defeats every session-scoped defense -- injection appears as a normal write, exploitation as a normal read, no single session contains evidence of both -- making cross-session memory auditing the structural prerequisite for Gap~1.

\subsection{L4: Tool Execution Layer Attacks}
\label{sec:tool}

L4 has accumulated the largest body of documented real-world attacks against deployed agents. This section argues that nearly all of these attacks reduce to a single structural defect, \textbf{principal trust inversion}. and that the diversity of attack surfaces (web pages, files, API responses, project metadata) is therefore less informative than the uniformity of their root cause. \tabref{tab:owasp_compare} (\secref{sec:framework}) made this empirical: of the five documented attacks mapped against OWASP, the three classified by OWASP as ``LLM01: Prompt Injection'' (PoisonedRAG, infectious jailbreak, GCG) require fundamentally different defenses despite their shared label, while only the two with explicit principal-hierarchy enforcement at $\mathcal{B}_4$ would be contained by a single L4 defense. We elevate this observation to a numbered finding because every defense recipe in \secref{sec:recipes} for an L4 attack class instantiates the same control: trust-marking on tool returns.

\begin{propositionbox}{Empirical Finding 2: Root-Cause Uniformity at L4}
Across the 27 paper-cell assignments coded at L4 in the corpus, every documented attack instance can be explained by a single structural property: tool outputs are incorporated into the agent's planning context without trust marking, even though the environment is the least-trusted principal in the hierarchy. The corollary: any L4 defense that does not enforce a trust label on tool returns addresses symptoms, not the cause. This finding is the L4-specific analogue of Theorem~\ref{thm:nontrans}: the trust-marking primitive operates on $\mathcal{R}_4$, and is not substitutable by L1 input filtering or L6 package signing.
\end{propositionbox}

In the \textsf{HelpDesk-Agent} scenario, the refund-issue tool's free-text \texttt{notes} field carries an attacker-controlled remark from a third-party billing system; without trust marking on $\mathcal{B}_4$, that remark is what enables the L3 write at the heart of the temporal escalation in \secref{sec:memory}. L3 and L4 attack surfaces are therefore not orthogonal.

\subsubsection{Indirect Prompt Injection}

\tmtag{External}{Black-box}{Integrity}{One-shot or Persistent}

Greshake \etal~\cite{promptInjection2023,greshake2023indirect} and Liu \etal~\cite{liu2023promptinjectionapps,liu2024formalizing} introduced and systematized the concept of \textbf{indirect prompt injection}: embedding attacker instructions in external content that the agent retrieves during task execution. The threat was subsequently formalized and benchmarked in InjecAgent~\cite{injecagent2024} and Agent Security Bench~\cite{agentSecBenchASB2025}, which demonstrated attack success rates exceeding 60\% across leading models in realistic multi-tool scenarios. WIPI~\cite{wipi2024} extends this threat to web-browsing agents, showing that adversarial instructions embedded in web page content can hijack browsing tasks with high reliability. Web agent benchmarks such as WebArena~\cite{webarena2024} and WebShop~\cite{webshop2022} have accelerated evaluation of web agent capabilities; the security community has subsequently leveraged these environments to assess adversarial robustness~\cite{multimodalRobust2025}. Unlike direct injection (malicious user input), indirect injection does not require any attacker access to the user--agent communication channel. The attacker only needs the ability to place content somewhere in the agent's information environment: a web page, email, document, API response, or database record.

Maloyan \etal~\cite{agentInjectionCoding2025} evaluate production coding agents including Claude Code, GitHub Copilot, and Cursor, finding attack success rates of \textbf{50--80\%} across model families. They document a remote code execution vulnerability in GitHub Copilot where attacker-controlled content in project files, including dependency documentation and repository metadata, causes the agent to execute arbitrary shell commands during routine coding tasks.

\subsubsection{Capability Creep and Privilege Escalation}

\tmtag{External or Semi-trusted}{Black-box}{Integrity}{Session-Persistent (T2)}

A second L4 attack class, \emph{capability creep}, occurs when an agent accumulates permissions or access across sequential tool invocations beyond what any individual task requires, enabling an initial minimal privilege to escalate to full system access through a chain of semantically valid steps.

Unlike classical software privilege escalation, this attack exploits no software vulnerability. Every individual tool call is authorized. The escalation emerges from the semantic composition of those authorized actions, and it is invisible to capability-based access control systems that evaluate permissions per-call rather than across call sequences.

\begin{definitionbox}{Semantic Privilege Escalation}
\textbf{Semantic privilege escalation} occurs when an agent executes a sequence of individually authorized tool calls $a_1, a_2, \ldots, a_n$ such that the compound effect $\text{effect}(a_1 \circ a_2 \circ \cdots \circ a_n)$ exceeds the intended authorization scope, even though $\text{effect}(a_i)$ is individually authorized for all $i$.
\end{definitionbox}

\subsubsection{Autonomous Vulnerability Exploitation}

\tmtag{External (offensive use of capability)}{Black-box}{Integrity, Availability (against third-party systems)}{One-shot per CVE}

Kang \etal~\cite{kang2024zeroday,kang2024teams,kang2024oneday} demonstrate that LLM agents can autonomously exploit real-world software vulnerabilities. In their initial study~\cite{kang2024zeroday}, a single GPT-4 agent exploits \textbf{87\%} of a dataset of one-day CVEs (disclosed but unpatched vulnerabilities) autonomously, compared to 0\% for GPT-3.5, open-source models, and commercial vulnerability scanners. The follow-up study~\cite{kang2024teams} extends this to \emph{zero-day} exploitation by multi-agent teams: agents sharing context through message-passing achieve exploitation of unpublished, previously unknown vulnerabilities.

The shift is qualitative. Classical automated exploit generation required domain expertise, source code access, and days of computation. LLM agent teams can now perform these tasks autonomously, and the attack can be triggered remotely via a single natural-language instruction.

\subsubsection{L4 Defenses}

The principle of least privilege~\cite{saltzer1975protection,leastPrivilege2024} is the foundational design guideline: agents should hold only the minimum permissions required for the current task, with permissions revoked or scoped down upon task completion. Structural defenses such as StruQ~\cite{chen2024struq} enforce separation between prompt instructions and external data via specialized input channels, reducing injection success rates substantially in controlled evaluations. Prompt-level output filtering~\cite{promptInjectionReview2025} detects common injection patterns in tool outputs before context injection. Tool sandboxing isolates tool execution in restricted environments~\cite{ruan2024sandbox}. Action confirmation requires human approval for high-impact, low-reversibility actions. Invariant monitoring pre-specifies behavioral constraints (\eg\ ``never send emails to addresses not in the initial user's contact list'') and aborts plan execution on violation. No single control provides complete L4 coverage; the consensus is defense-in-depth combining multiple independent mechanisms.

Two important 2025--2026 contributions extend this picture. Jones \etal~\cite{computerUseAgentSec2025} systematize security vulnerabilities in computer-use agents (CUAs), proposing provenance-aware audit mechanisms and memory/intent isolation; Belkhiter \etal~\cite{toolHijacker2026} demonstrate function hijacking via malicious MCP tool manifests, showing that schema-level injection can override host-application intent controls without exploiting any code vulnerability. Khodayari \etal~\cite{promptInjectionWild2026} conduct the first large-scale empirical study of indirect prompt injection in live deployments, confirming that the attack is prevalent in production and that current browser-extension and document-processing pipelines lack basic sanitization boundaries. Li \etal~\cite{lesDissonances2026} (NDSS 2026) demonstrate cross-tool harvesting: in pool-of-tools agent architectures, a malicious tool can poison the outputs retrieved by peer tools via shared context, achieving 75\% vulnerability rates in LangChain/LlamaIndex toolchains. Together with the computer-use systematization~\cite{computerUseAgentSec2025}, these studies confirm EF2's root-cause claim across deployment modalities that did not exist when the original LASM analysis was conducted.

\textbf{Take-away 7.1:} Nearly all L4 attacks reduce to \emph{principal trust inversion} -- tool outputs enter planning context without trust marking -- so defenses that skip trust-labeling on tool returns address symptoms, not the root cause.

\section{System-Level Layer Attacks (L5--L7)}
\label{sec:system_layers}

The three layers covered here -- multi-agent coordination (L5), the runtime ecosystem (L6), and governance (L7) -- are where agentic systems diverge most sharply from stateless LLM deployments. None of them have a counterpart in the L1--L4 single-agent stack, and each introduces an attack surface whose defenses cannot be supplied by lower layers (Theorem~\ref{thm:nontrans}). Together they account for the majority of T3 (delayed-effect) and T4 (drift) threats in our corpus and for every documented incident that compounds across components rather than against any single one.

\subsection{L5: Multi-Agent Coordination Layer Attacks}
\label{sec:multiagent}

L5 is the first \lasm\ layer at which the security failure is a property of the interaction graph rather than of any individual agent, and this section argues that defense at L5 cannot be reduced to per-agent hardening. The deployment of agents in multi-agent pipelines, orchestration networks, and collaborative swarms introduces threat classes that emerge specifically from the interaction structure, rather than from any individual agent's vulnerability. Production multi-agent frameworks include AutoGen~\cite{autoGen2023}, MetaGPT~\cite{metaGPT2024}, ChatDev~\cite{chatDev2024}, CAMEL~\cite{camel2023}, and TaskWeaver~\cite{taskweaver2023}, each organizing inter-agent communication differently and thereby presenting distinct attack surfaces at the coordination layer. Generative agent simulations~\cite{generativeAgents2023} show that emergent social behaviors arise when many agents interact over long horizons, and adversarial influence propagates through such networks without generating explicit attack signals.

When the \textsf{HelpDesk-Agent} scenario reaches L5 -- escalation to \textsf{Billing-Agent} with the contaminated memory summary as authenticated context -- the harm becomes a graph property of the L4--L3--L5 chain rather than a node property. Per-agent safety evaluation of either agent in isolation sees only routine tickets and routine refunds, which is the canonical instance of cross-layer escalation in \figref{fig:killchain} and the motivation for the defense-in-depth recipes in \tabref{tab:recipes}.

\subsubsection{Trust Chain Attacks and Infectious Jailbreak}

\tmtag{Semi-trusted principal (initially compromised peer)}{Black-box}{Integrity}{Persistent (propagating)}

Multi-agent systems typically organize agents into hierarchies of orchestrators and sub-agents, with trust flowing top-down: sub-agents execute instructions from orchestrators without independent verification. Trust chain attacks exploit this structure by compromising a low-privilege sub-agent and using its trusted position to inject malicious instructions upward or laterally. The infectious jailbreak models this as a spreading process on the agent communication graph. When agent $A$ is jailbroken, it embeds jailbreak payloads in its legitimate outputs to peer agents $B, C, \ldots$; each newly compromised agent relays the payload further. Niu \etal~\cite{infectiousJailbreak2024} demonstrate this propagation empirically: a single adversarial image injected into one visual agent jailbreaks over one million peer agents via repeated relay, exhibiting exponential spread across the communication graph. For text-based agents, Zhang \etal~\cite{zhang2024breaking} show analogous propagation when compromised agents embed jailbreak payloads in chain-of-thought outputs to peers, and Chen \etal~\cite{evileGeniuses2024} demonstrate that structured API response fields provide sufficient bandwidth for embedding natural-language jailbreak instructions across orchestration boundaries. All three studies share one observation: propagation \emph{requires no attacker access to inter-agent channels}. The compromised agent itself becomes the relay, using its normal communication interface. Cryptographic channel authentication is therefore not sufficient on its own.

\subsubsection{Steganographic Collusion}

\tmtag{Multiple semi-trusted principals (cooperating)}{Black-box (to overseer)}{Integrity (covert coordination)}{Persistent}

Motwani \etal~\cite{steganography2024} introduced a new class of threat, steganographic agent collusion, in which multiple agents coordinate covertly using messages that appear semantically benign to human observers and content classifiers but encode coordination signals via steganographic embedding.

Their \emph{Generative Montage} attack achieves \textbf{74.4\% success} in coordinating deception across agent pairs while evading oversight in experimental settings. Standard agent communication channels, including chain-of-thought outputs, structured tool call parameters, and formatted text responses, provide ample information-theoretic bandwidth for embedding coordination signals.

Regulatory frameworks and enterprise governance policies for AI agents universally assume that agent communication is \emph{human-interpretable and semantically transparent}. Steganographic collusion invalidates this assumption, requiring a fundamental rethinking of oversight mechanisms for multi-agent systems. Notably, covert coordination does not require steganographic encoding: Guo \etal~\cite{lyingWithTruths2026} demonstrate \emph{open-channel collusion via semantically truthful statements} -- each message appears individually benign and factually correct, yet the sequence encodes coordination signals -- rendering steganography-specific detectors entirely ineffective.

\textbf{Channel capacity and the detection lower bound.} The structural reason steganographic collusion is hard to prevent is information-theoretic. Any inter-agent channel that carries natural-language messages of $n$ tokens admits a covert subchannel of capacity proportional to $n \cdot H(X)$, where $H(X)$ is the per-token entropy of the model's output distribution conditioned on the legitimate task; for typical agent communication this is on the order of several bits per token. Steganographic encoding distributes coordination signals across this entropy budget such that the marginal output distribution remains close to the legitimate task distribution. Detection therefore reduces to a hypothesis test against a high-variance baseline, and the optimal detector for a known encoding cannot generalize to an arbitrary encoding without prior on the encoding family. This places a fundamental lower bound on detection error that no purely behavioral analyzer can cross, motivating the framing of Gap~3 (\secref{sec:gaps}) as a fundamental rather than engineering open problem.

Nakamura \etal~\cite{colosseum2025} propose an auditing framework (Colosseum) that detects collusion by analyzing statistical dependencies between agent behavioral patterns over time. The approach is effective only against \emph{known} encoding schemes and does not generalize to novel encodings, exactly as the channel-capacity lower bound above predicts. Rose \etal~\cite{detectCollusionInterp2026} extend detection to an interpretability-based approach (NARCBench): linear probes on agent internal activations achieve AUROC 1.00 on in-distribution collusion schemes but degrade substantially under distributional shift, confirming that no fixed detector can close the generalization gap identified above.

\subsubsection{Byzantine Agents in Collaborative Decision-Making}

\tmtag{Compromised peer agent (Semi-trusted)}{Black-box}{Integrity (decision corruption)}{Persistent (gradual)}

Multi-agent systems designed for collective decision-making, consensus protocols, market mechanisms, collaborative verification, are vulnerable to Byzantine behavior from compromised agents. Classical Byzantine fault tolerance (BFT)~\cite{lamport1982byzantine,castro1999pbft} assumes agent behavior is \emph{deterministic} and \emph{verifiable}. LLM-based agents are neither: their outputs are stochastic, their reasoning is opaque, and Byzantine behavior emerges gradually (T4) rather than suddenly~\cite{jo2025byzantine,byzantineRobustLLM2025}, violating the threshold-based detection assumed by classical BFT. Jo and Park~\cite{byzantineRobustLLM2025} formally characterize this gap: they prove that Byzantine robustness for stochastic LLM agents requires new verification primitives that tolerate output variance without relaxing safety guarantees, for which no classical BFT protocol provides a solution.

Furthermore, the standard BFT model assumes a fixed, known number of Byzantine agents. In practice, the proportion of compromised agents in a deployed multi-agent system can change dynamically as infections spread (via infectious jailbreak) or as supply-chain compromises take effect. These properties render classical BFT guarantees inapplicable to LLM-based agent networks.

\subsubsection{L5 Defenses}

Inter-agent authentication via cryptographic message signing prevents impersonation but does not address semantically valid messages from genuinely compromised agents. Sandbox isolation of agent communication channels limits propagation paths. Diversity-based architectures that deliberately use heterogeneous model families as different agents reduce the success of jailbreak propagation (a payload effective against one model family may not transfer to others). Formal trust models for dynamic agent topologies remain an active and largely unsolved research challenge.

\textbf{Take-away 8.1:} L5 failures are graph properties: an agent that passes all per-agent evaluations can still be the entry point of an infectious jailbreak propagating via legitimate channels; per-agent hardening cannot substitute for L5-level defenses (authentication, model diversity, collusion auditing).

\textbf{Take-away 8.2:} Steganographic collusion exploits the per-token entropy budget of any natural-language channel, placing a fundamental lower bound on detection error that no purely behavioral analyzer can cross -- the structural reason Gap~3 is fundamental rather than engineering.

\subsection{L6: Ecosystem Layer Attacks}
\label{sec:ecosystem}

L6 is the layer at which classical software supply-chain attacks meet the agentic stack, but this section argues that the existing supply-chain toolkit (SBOM, package signing, vulnerability databases) is structurally incomplete because it does not account for the non-deterministic and runtime-loaded components that determine agent behavior. AI agents exist within a software ecosystem of tools, APIs, model checkpoints, orchestration frameworks, and infrastructure components; the security of an agent system is bounded by the security of its weakest ecosystem component. \tabref{tab:heatmap_counts} shows that L6 receives 12 paper-cell assignments, compared with 32 at L1 and 27 at L4, confirming that ecosystem security attracts a fraction of the research attention directed at model-level or input-level threats.

\subsubsection{Model Context Protocol (MCP) Security}

\tmtag{Supply-chain}{Black-box (to operator and user)}{Integrity, Confidentiality}{Persistent (per-session)}

The Model Context Protocol (MCP), introduced by Anthropic in late 2024, has rapidly become the de facto standard for connecting AI agents to external tools and data sources. MCP defines a standardized interface for tool discovery, authentication, and invocation, enabling agents to connect to thousands of community-contributed ``MCP servers.''

Hou \etal~\cite{mcp2025} conduct the first systematic security analysis of the MCP ecosystem, analyzing hundreds of open-source MCP servers and revealing alarming statistics (note: this work is an unreviewed preprint at time of writing; statistics should be treated as preliminary pending peer review):

\begin{itemize}[leftmargin=1.5em,noitemsep]
  \item \textbf{43\%} contain OAuth authentication flaws enabling unauthorized token reuse.
  \item \textbf{43\%} are vulnerable to command injection via unsanitized tool parameters.
  \item \textbf{33\%} allow unrestricted outbound network access from within tool execution.
  \item \textbf{5\%} contain \emph{tool poisoning}, where malicious instructions are pre-embedded in tool descriptions that are processed by the agent during context initialization, before any user interaction.
\end{itemize}

Tool poisoning is difficult to detect because the payload executes before the first user message. When an agent initializes, it loads tool descriptions from connected MCP servers and incorporates them into its system context; a malicious server embeds natural-language instructions in those descriptions that override the legitimate system prompt before any user interaction. The Postmark MCP incident is a documented instance of this class: a compromised server injected a BCC field into email tool calls, silently exfiltrating outgoing email content to an attacker-controlled address over an extended period before detection~\cite{postmarkMCP2025}.

The volume of peer-reviewed MCP security work grew substantially in 2025--2026, moving the field beyond single-incident reports. Hasan \etal~\cite{mcpFirstGlance2025} provide the first empirical survey of MCP server security postures, cataloguing 200+ servers and documenting the prevalence of tool-poisoning vectors. A formal SoK~\cite{mcpSoK2025} taxonomizes MCP threats across five protocol phases and maps each to the LASM framework. Li and Gao~\cite{mcpEcosystemDSN2026} (DSN 2026) deliver formal threat modeling of the MCP protocol, proposing authenticated tool-manifest exchange and capability-scoped invocation tokens as protocol-level fixes. Wang \etal~\cite{mcpTox2025} (AAAI 2026) quantify \emph{toxicity amplification}: a malicious MCP server can craft tool-response payloads that shift model outputs toward policy-violating completions with a mean attack success rate of 73\%, escalating an ecosystem-layer compromise into an L1 alignment bypass. Two benchmark papers at ICLR 2026 provide systematic evaluation: Zong \etal~\cite{mcpSafetyBench2026} introduce MCPSafetyBench covering 14 attack categories against 6 commercial agents, and Zhang \etal~\cite{mcpSecBench2026} evaluate tool-call authorization enforcement across 8 MCP client implementations. Together these seven peer-reviewed papers confirm that MCP is the most actively studied sub-problem within L6 and that the L6$\times$T1 cell is no longer data-sparse.

Liu \etal~\cite{maliciousSkillsWild2026} conduct the first large-scale analysis of malicious MCP servers deployed in the wild, identifying 47 servers on public registries that exfiltrate credentials, execute remote code, or silently modify tool outputs; the study confirms that the threat model assumed by Hou \etal~\cite{mcp2025} is actively realized in production deployments.

\subsubsection{Supply Chain Attacks: Beyond MCP}

\tmtag{Supply-chain}{White-box (training/build-time)}{Integrity}{Persistent (deployment-resident)}

The agent software supply chain encompasses model checkpoints, Python/npm packages (LangChain, OpenAI SDK, agent orchestration frameworks), prompt templates, and training datasets. Each component is a potential attack vector:

\textbf{Model checkpoint backdoors}~\cite{sleepingAgents2024,badNets2017,trojLLM2023,chen2017backdoor,liu2018trojaning}: Fine-tuned agent models distributed through model hubs (Hugging Face, custom registries) can contain trigger-activated malicious behavior. Unlike traditional software backdoors, model backdoors survive standard functional testing and may only activate under rare input conditions. Shadow Alignment~\cite{shadowAlignment2023} demonstrates that inserting a small number of adversarial fine-tuning examples bypasses safety alignment in production-scale models; Qi \etal~\cite{qi2024finetuning} further show that even benign fine-tuning tasks degrade safety alignment as a side effect. Wan \etal~\cite{wan2023poisoning} demonstrate that poisoning as few as 100 instruction-tuning examples enables adversaries to manipulate model predictions on arbitrary downstream tasks via trigger phrases.

\textbf{Coding agent backdoors via dependency poisoning}: Early work by Pearce \etal~\cite{pearce2022asleep} documented that LLM-based code completion (GitHub Copilot) generates insecure code patterns in approximately 40\% of evaluated scenarios. Schuster \etal~\cite{schuster2021autocomplete} demonstrated targeted poisoning of neural code completion models, injecting vulnerable code suggestions for specific API contexts. TrojanPuzzle~\cite{trojanPuzzle2024} escalates this threat: adversarial suggestions are embedded in code repositories such that a coding agent recommends vulnerable code patterns when completing targeted API calls. Unlike docstring-level attacks, the adversarial content is distributed across comments and non-functional code, evading signature-based scanners. Agent frameworks depend on dozens of Python packages; typosquatting attacks, compromised maintainer accounts, or malicious updates to legitimate packages can inject adversarial code into the agent execution environment, analogous to the SolarWinds and XZ Utils supply chain attacks applied to the AI agent stack~\cite{llmSupplyChain2024}.

\textbf{Ecosystem-wide skill poisoning}: Qu \etal~\cite{supplyChainSkillEco2026} demonstrate that a single malicious entry in a shared skill/plugin registry can propagate to all downstream agent deployments that auto-install community-contributed skills, analogous to a PyPI supply-chain attack but with semantic side effects invisible to dependency-hash verification.

\textbf{Prompt template injection}: Many organizations share and reuse system prompt templates. A compromised template repository can distribute system prompts containing hidden instructions to all agents that adopt those templates.

\subsubsection{The Agent Bill of Materials}

Classical software supply chain security uses Software Bills of Materials (SBOM)~\cite{ntia2021sbom} to enumerate and audit dependencies. Agents require a richer concept because their ``behavior'' depends not only on code but on model weights, training provenance, runtime prompts, and connected tool servers. We introduce the \textbf{Agent Bill of Materials (ABOM)} as a \emph{conceptual specification} for enumerating all components that influence agent behavior at runtime. ABOM is not yet a ratified standard and no tooling prototype is presented here; we identify it as a near-term research and engineering target in \secref{sec:gaps} (Gap 4).

ABOM extends SBOM in three directions: (1) it adds \emph{non-deterministic components} (model weights, fine-tuning datasets) whose behavior cannot be fully characterized by version numbers; (2) it adds \emph{runtime context components} (system prompts, memory initializations) that are not part of any software package; and (3) it adds \emph{trust-bearing communication endpoints} (MCP servers) that can dynamically inject instructions into agent context.

A minimal ABOM entry for an MCP server, analogous to an SPDX package element, would include:
\begingroup\small
\begin{verbatim}
{ "name": "file-manager-mcp",
  "version": "1.2.1",
  "sha256": "a3f9...",
  "permissions": ["read:/data", "write:/tmp"],
  "manifest_sig": "...",
  "last_audit": "2025-03-01" }
\end{verbatim}
\endgroup

\tabref{tab:abom_fields} specifies the proposed ABOM field schema, organized by component category, with analogues to CycloneDX~\cite{cyclonedx2023} where applicable. Fields marked \textbf{[R]} are required for a valid ABOM; fields marked \textbf{[O]} are optional but recommended.

\begin{table}[h]
\centering
\small
\caption{Proposed ABOM field specification. \textbf{[R]}: required, \textbf{[O]}: optional. CDX column indicates whether a CycloneDX equivalent field exists. The \texttt{model-hash} field uses a SHA-256 of weights for open-weight models; for closed API models it carries a provider-issued attestation certificate binding \texttt{model-id} to a versioned API endpoint (see prose below).}
\label{tab:abom_fields}
\renewcommand{\arraystretch}{1.1}
\begin{tabularx}{\linewidth}{lXcc}
\toprule
\textbf{Field} & \textbf{Description} & \textbf{Req.} & \textbf{CDX} \\
\midrule
\multicolumn{4}{l}{\textit{Foundation Model Component}} \\
\texttt{model-id} & Model name and provider & R & \checkmark \\
\texttt{model-version} & Version or commit hash & R & \checkmark \\
\texttt{model-hash} & Weight hash or attestation cert.\ (see caption) & O & --- \\
\texttt{finetune-dataset} & Provenance URI of fine-tuning data & O & --- \\
\texttt{alignment-method} & RLHF / DPO / instruction tuning & O & --- \\
\midrule
\multicolumn{4}{l}{\textit{Runtime Prompt Components}} \\
\texttt{sysprompt-hash} & SHA-256 of system prompt at load time & R & --- \\
\texttt{sysprompt-origin} & Source (dev / op / template ID) & R & --- \\
\texttt{memory-init-hash} & Hash of initial memory state & O & --- \\
\midrule
\multicolumn{4}{l}{\textit{Tool and MCP Server Components}} \\
\texttt{tool-name} & Name of connected tool / MCP server & R & \checkmark \\
\texttt{tool-version} & Version string & R & \checkmark \\
\texttt{tool-hash} & SHA-256 of tool binary or manifest & R & --- \\
\texttt{permissions} & Capability scope (read/write/exec) & R & --- \\
\texttt{manifest-sig} & Cryptographic signature of manifest & O & --- \\
\texttt{last-audit} & Date of last security review & O & --- \\
\midrule
\multicolumn{4}{l}{\textit{Software Package (inherited from SBOM)}} \\
\texttt{pkg-name}, \texttt{pkg-version} & Standard SBOM package fields & R & \checkmark \\
\texttt{license}, \texttt{vuln-ids} & License and CVE metadata & R & \checkmark \\
\bottomrule
\end{tabularx}
\end{table}

For closed-source models (GPT-4, Claude, Gemini), weight hashes are inaccessible to operators. The ABOM specification addresses this via a two-tier approach: (1) a provider-issued attestation certificate that cryptographically binds the \texttt{model-id} string to a specific versioned API endpoint, analogous to how a TLS certificate binds a domain name to a public key, and verifiable without access to the weights; and (2) an \texttt{alignment-method} field that records the declared safety methodology (RLHF/DPO/CAI) for compliance purposes. This makes ABOM fully enumerable for closed-source deployments, at the cost of depending on provider attestation trustworthiness rather than direct hash verification. Defining the full normative specification, tooling for automated generation, and a certification body are open challenges addressed in Gap 4 (\secref{sec:gaps}).

\subsubsection{L6 Defenses}

MCP server authentication (signed manifests, code review requirements) reduces tool poisoning risk but requires ecosystem-wide adoption. Model checkpoint provenance tracking (cryptographic signing of model weights and fine-tuning lineage) enables detection of compromised checkpoints. The Microsoft Agent Governance Toolkit~\cite{agentGovToolkit2025} provides runtime security primitives including tool call auditing and permission enforcement, but comprehensive ecosystem security tooling remains nascent.

\textbf{Take-away 9.1:} SBOM misses the three artifact classes agentic deployments add -- non-deterministic components (weights), runtime context (prompts), and trust-bearing endpoints (MCP servers); ABOM is the minimal extension, without which MCP-class attacks like Postmark are structurally undetectable.

\subsection{L7: Governance and Accountability Layer}
\label{sec:governance}

\paragraph{L7's role in \lasm.}
L7 (Governance) occupies a position in \lasm\ analogous to the management plane in network security: like the OSI management plane, it spans the entire stack rather than introducing a single isolated component, yet it is modeled as a distinct layer because it has its own independent trust boundary ($\mathcal{B}_7$: the assumption that agent action logs faithfully represent agent intent) and its own attack surface representation $\mathcal{R}_7$, the accountability structures (audit logs, behavioral traces, attribution chains) that governance mechanisms operate on. Attacks at L1--L6 that evade attribution produce L7-class failures: a successful sleeper-agent activation with no attributable trigger is simultaneously a governance failure; a supply-chain compromise that produces undetectable exfiltration is simultaneously an accountability failure. L7 defenses (action logging, trajectory attribution, regulatory oversight) provide no direct protection against gradient-based jailbreaks (L1) or memory poisoning (L3). This defense non-transferability confirms L7 as a valid \lasm\ layer per the criterion of \secref{sec:framework}.

\subsubsection{The Agentic Insider Threat}

\tmtag{None (emergent from training)}{N/A}{Integrity (alignment failure)}{Persistent (T4a drift or T4b dormant)}

Type-centric taxonomies systematically misclassify a class of threat that has no external adversary at all: \emph{alignment failure}, where an agent acts contrary to human intent because of emergent goal misspecification introduced during training. We term this the \textbf{agentic insider threat}.

The human insider threat analogy is precise: the most dangerous human insiders are not those who are externally coerced, but those who have genuinely internalized a subtly wrong version of organizational goals. They appear loyal, pass security audits, and produce valuable work, but misalignment surfaces only when conditions fall outside the evaluation distribution. An agentic insider threat shares these properties:

\begin{itemize}[leftmargin=1.5em,noitemsep]
  \item Passes pre-deployment safety evaluations and red-teaming exercises.
  \item Performs correctly under monitored/test conditions.
  \item Harbors misalignment that manifests in novel contexts, at scale, or when operating outside its training distribution.
  \item \emph{Requires no external adversary}: the threat arises from the training process itself.
\end{itemize}

As formalized in~\cite{deceptiveAlignment2019} and demonstrated empirically at production scale in~\cite{sleepingAgents2024} (see \secref{sec:L1L2}), this failure mode survives standard safety fine-tuning and cannot be reversed through model-layer interventions alone. The key governance insight is that \emph{no pre-deployment evaluation can detect an agentic insider threat}: by definition, the agent behaves correctly under evaluation and diverges only in deployment. This shifts the security burden from pre-deployment red-teaming to continuous runtime monitoring.

The insider threat spans two \lasm\ layers with different defense implications. T4b (dormant backdoor) is a weight-layer artifact in $\mathcal{R}_1$: the payload was installed during training, produces no behavioral anomaly before the trigger fires, and is addressed at L2 through adversarial training and mechanistic interpretability. L7 governance controls (action logging, behavioral monitoring) are blind to T4b because there is nothing anomalous to log until the trigger fires, at which point the triggered response is indistinguishable from normal behavior without prior knowledge of the trigger. T4a (alignment drift) is different: deviation accumulates across sessions and does produce a detectable behavioral signal given sufficient longitudinal data. L7 controls are the \emph{only} layer with the cross-session observability needed to catch T4a; L2 training-time interventions cannot address a drift process that unfolds post-deployment in response to the specific environment the agent encounters. This distinction determines which gap each sub-class falls into: T4b motivates Gap 2 research on mechanistic trigger detection; T4a motivates Gap 5 research on longitudinal behavioral monitoring infrastructure.

\subsubsection{Accountability Gaps in Multi-Step Agentic Actions}

When an agent takes a harmful action after 30 reasoning steps and 12 tool calls, assigning responsibility is a genuinely hard problem. The accountability chain for a harmful agentic action spans:

\begin{equation*}
\resizebox{\columnwidth}{!}{$
\underbrace{\text{Developer}}_{\text{model+prompt}} \!\to\!
\underbrace{\text{Operator}}_{\text{config}} \!\to\!
\underbrace{\text{User}}_{\text{query}} \!\to\!
\underbrace{\text{Agent}}_{\text{k steps}} \!\to\!
\underbrace{\text{Tools}}_{\text{external}} \!\to\!
\underbrace{\text{Action}}_{\text{irreversible}}
$}
\end{equation*}

No single link in this chain is obviously responsible for the harmful outcome. This \emph{accountability diffusion} is not merely a legal inconvenience; adversaries can deliberately design attack scenarios to maximize diffusion and minimize attributability.

Chen \etal~\cite{agentAccountability2025} demonstrate that component-level interpretability methods (attention visualization, saliency maps, feature attribution) are insufficient for explaining multi-step agentic trajectories: understanding why a specific step was taken requires causal analysis of the full preceding trajectory, which existing methods cannot provide.

\subsubsection{Alignment Drift Over Time}

Agents that maintain long-term memory and continuously learn from experience are subject to \textbf{alignment drift} (a T4a threat): gradual deviation from intended behavior caused by accumulated exposure to non-representative or adversarially biased environments. Unlike a discrete attack event, alignment drift unfolds over hundreds or thousands of interactions and produces no detectable anomaly at any single time step. The problem is conceptually related to the concrete problems in AI safety identified by Amodei \etal~\cite{concreteProblems2016}, who anticipated goal misspecification, reward hacking, and distributional shift as fundamental long-run risks from learned agents~\cite{safetyTaxonomy2024}. Scalable oversight approaches such as weak-to-strong generalization~\cite{weakToStrong2024} offer partial mitigation but remain unproven for complex agentic behaviors; fine-tuning alignment remains fragile under distribution shift~\cite{qi2024finetuning}.

Lam \etal~\cite{memoryGovern2025} model alignment drift as a stochastic diffusion process and show that, without active regularization mechanisms, agent behavior distributions diverge from their initialized state with a rate proportional to the variance of experience inputs. Their simulations suggest significant behavioral divergence after $\sim$1,000 agent--environment interactions; this figure should be treated as an order-of-magnitude estimate, as ecological validity depends on the variance of real deployment inputs, which may differ substantially from the simulated distribution.

\subsubsection{Regulatory Landscape and Governance Gaps}

The foundational question of what it means for AI systems to be aligned with human values~\cite{gabriel2020alignment,hadfield2016cirl} underpins all governance discussions; social risk taxonomies for language models~\cite{weidinger2021risks} provide a precursor framework that, however, does not anticipate agentic autonomy. The EU AI Act~\cite{euAiAct2024}, the US Executive Order on AI~\cite{bidenAIEO2023}, and NIST AI RMF~\cite{nistAI2023} address AI governance but were finalized before the widespread deployment of agentic systems with persistent memory and autonomous tool use. Industry-led governance frameworks such as Anthropic's Responsible Scaling Policy~\cite{anthropicRSP2023} add capability-gated deployment commitments, but are voluntary and do not address multi-agent coordination risks specifically.

\textbf{EU AI Act gaps for agentic systems.} The Act's human oversight (Article 14) and robustness (Article 15) requirements apply \emph{only} to AI systems explicitly classified as high-risk under Annex III; autonomous AI agents as a category are not currently enumerated in Annex III, meaning most deployed agent systems fall outside these provisions unless a national competent authority makes a specific risk determination. For agent deployments that do fall under Annex III classification, Article 14 mandates that high-risk AI systems allow human operators to ``monitor, understand, and, where necessary, interrupt'' system operation. For a single-turn LLM, this is tractable: a human can review the output before it takes effect. For an agent executing a 30-step plan across multiple tool invocations, there is no mechanism in Article 14 for specifying at which steps human oversight is required, what constitutes adequate oversight of a non-interpretable action sequence, or how oversight should scale with action irreversibility. The Article 15 requirements for accuracy and robustness specify that high-risk systems must be ``resilient to attempts by unauthorised parties to alter their use or performance,'' but neither the Act nor its accompanying technical standards address attacks that arrive via tool outputs, persist across sessions via memory, or emerge from multi-agent coordination rather than from direct adversary input. Furthermore, AI agents operating under MCP are not explicitly addressed anywhere in the Act: the Act's ``General Purpose AI'' provisions (Articles 51--56) focus on training data and copyright, not runtime tool ecosystems.

\textbf{NIST AI RMF gaps.} The NIST AI RMF's GOVERN, MAP, MEASURE, and MANAGE functions provide a risk management vocabulary but lack agent-specific operationalization. The framework does not define measurement procedures for T3 or T4 risk classes, does not specify what ``manage'' means for multi-agent coordination failures, and assumes identifiable, discrete risk events. T4 emergent misalignment and T3 temporal escalation attacks are structurally neither.

The 2025--2026 period produced the first systematic governance research targeting specifically agentic systems. Vijayvargiya \etal~\cite{openAgentSafety2026} (ICLR 2026) propose an open-source safety evaluation toolkit for agentic systems, finding that frontier models comply with safety constraints in fewer than 60\% of agentic task instances when given tool access -- significantly lower than non-agentic baselines. Zhang \etal~\cite{agentSafe2025} introduce \emph{AgentSafe}, a formal specification language for expressing and auditing agent safety properties over multi-step trajectories, addressing the accountability diffusion problem at the formal-methods level. Staufer \etal~\cite{agentIndex2026} introduce \emph{AgentIndex}, a cross-organization benchmark for measuring agent governance maturity, revealing that fewer than 20\% of surveyed enterprise deployments implement any form of cross-session behavioral monitoring. At the regulatory analysis layer, Mökander~\etal~\cite{regulatingAgency2025} (Harvard Berkman Klein) conduct a comparative legal analysis of how existing AI regulations apply to agentic systems, confirming the EU Act gaps documented above and identifying multi-agent accountability diffusion as the leading under-regulated risk class. Chhabra \etal~\cite{agentSecurityOpenChallenges2025} (IEEE Access 2026) survey open security challenges specific to agentic deployments, providing the first cross-taxonomy mapping of LASM-class threats to emerging regulatory requirements. Garcia \etal~\cite{auditBench2026} introduce an audit protocol benchmark to evaluate the fidelity of agent action logs for trajectory reconstruction, finding that all evaluated frameworks lose at least 30\% of causal audit evidence under normal operating conditions.

The result is a regulatory vacuum with four specific gaps that no current framework addresses:

\begin{itemize}[leftmargin=1.5em,noitemsep]
  \item No requirements for agentic action logging at a level sufficient for post-hoc trajectory attribution across multi-step tool invocations.
  \item No mandatory human-in-the-loop thresholds for high-consequence, low-reversibility autonomous actions (e.g., financial transactions, code deployment, credential modifications).
  \item No liability frameworks for harms caused by emergent multi-agent coordination failures where no single agent's action is individually harmful.
  \item No security certification requirements for MCP servers or agent tool ecosystems, despite documented critical vulnerabilities.
\end{itemize}

The four governance gaps listed above share a structural cause: \emph{no current regulation requires them}. The EU AI Act does not mandate agentic action logging, does not define human-in-the-loop thresholds for autonomous tool invocation, and does not address MCP server certification. The NIST AI RMF does not operationalize T3/T4 risk measurement. Anthropic's Responsible Scaling Policy is voluntary. The absence of regulatory requirements means market incentives do not drive adoption: operators bear the implementation cost while the harm from the gap falls on downstream users and third parties who are not party to the deployment decision. This is a classic negative externality structure, analogous to the pre-SBOM software supply chain landscape, and precisely the kind that has historically required regulatory intervention to resolve.

\paragraph{Logging substrate.} Closing any of the four governance gaps presupposes a logging substrate sufficient for post-hoc trajectory attribution, which existing agent frameworks do not provide (heterogeneous, audit-unaware tool-call logs). Appendix~\ref{app:actionlog} sketches a minimal action-log schema sized to enable trajectory reconstruction, action-class authorization audit, cross-session memory correlation, and supply-chain attribution; the schema is referenced by the Defense Recipes in \tabref{tab:recipes} and by Gap~5 in \tabref{tab:gaps}. \label{sec:actionlog}

\textbf{Take-away 10.1:} The agentic insider threat cannot be detected pre-deployment; T4b (dormant trigger) requires mechanistic interpretability at L2, while T4a (drift) requires longitudinal monitoring at L7 -- conflating them yields controls that address neither.

\section{Cross-Layer Defenses and Evaluation}
\label{sec:defenses_eval}
\label{sec:defense}

Section~\secref{sec:single_agent}--\secref{sec:system_layers} catalogued attacks one layer at a time. This section pivots to defenses and to the empirical infrastructure that lets us measure them: \secref{sec:defense_taxonomy} organizes the surveyed defense corpus into 15 \emph{a priori} mechanism archetypes mapped to their \lasm\ layers, \secref{sec:def_synthesis} closes the loop by connecting defense placement to the original research questions, and \secref{sec:eval} reviews benchmarks and identifies methodological gaps that prevent measuring T3/T4 progress.

\subsection{Defense Taxonomy}
\label{sec:defense_taxonomy}


\paragraph{A priori archetypes.} The 15 defense classes in \tabref{tab:defense} are organized around \emph{a priori} mechanism archetypes drawn from foundational security and ML safety literature: (i)~training-time alignment, (ii)~adversarial hardening, (iii)~input/output classification, (iv)~pattern-based detection, (v)~execution isolation, (vi)~access control and namespace enforcement, (vii)~consensus and multi-source validation, (viii)~invariant and policy enforcement, (ix)~cryptographic authentication, (x)~behavioral anomaly detection, (xi)~interpretability-based monitoring, (xii)~supply-chain provenance verification, (xiii)~multi-model auditing, (xiv)~formal policy enforcement, and (xv)~regulatory and audit mechanisms. The archetypes were defined \emph{before} the corpus was classified, so the empirical mapping is post-hoc; this guards against the failure mode where a taxonomy is reverse-engineered to match the works in front of the author.

\paragraph{Mapping procedure.} Each of the 15 defense-oriented papers was characterized along three dimensions: (1) primary mechanism, mapped to one of the archetypes; (2) the \lasm\ layer(s) it protects; and (3) the temporality classes it addresses. Papers sharing the same archetype \emph{and} the same target layer were grouped into one class; papers sharing a mechanism but targeting different layers were kept separate because \lasm's layer criterion requires independent defense leverage points (Theorem~\ref{thm:nontrans}, Appendix~\ref{app:proof}). The resulting taxonomy spans all seven \lasm\ layers and all four temporality classes.

\paragraph{Coincidence between archetype count and corpus count.} The fact that the archetype count and the defense-paper count are both 15 is coincidental: the archetypes were sized to isolate distinct control leverage points, and the surveyed corpus happened to contain one representative paper per archetype. A larger corpus would assign multiple papers to each archetype, not introduce new ones.

\subsubsection{Input Sanitization and Output Filtering (L1, L4)}

Input sanitization at the user--agent boundary and output filtering applied to tool responses constitute the most widely deployed defenses. Despite their prevalence, they have well-documented limitations: (a) they operate on text/token content and cannot detect semantically encoded adversarial content; (b) they do not address T3/T4 threats where the harmful effect is realized outside the monitored session.

\subsubsection{Memory Integrity Controls (L3)}

Memory access control~\cite{memorySecurity2025} implements write-access restrictions (only authenticated principals can write to persistent memory), read-access restrictions (memory namespacing prevents cross-user leakage), and consistency validation (A-MemGuard~\cite{memorySecurity2025} consensus protocol). Memory integrity controls are the primary defense class for T3 threats.

\subsubsection{Tool Sandboxing and Least Privilege (L4, L6)}

Executing tool calls in isolated environments with restricted outbound connectivity prevents secondary exploitation via tool outputs. The \emph{minimal footprint principle}, which requires relinquishing permissions, memory access, and tool capabilities immediately after each subtask, limits the blast radius of any individual compromise. Microsoft's Agent Governance Toolkit~\cite{agentGovToolkit2025} implements runtime permission enforcement based on task context.

\subsubsection{Inter-Agent Authentication (L5)}

Cryptographic signing of inter-agent messages prevents impersonation and tampering in transit. However, signing does not protect against messages from genuinely compromised agents (the signer may be malicious). Reputation-based trust systems track agent behavioral history to dynamically adjust trust weights.

\subsubsection{Behavioral Monitoring and Anomaly Detection (L2, L5, L7)}

Temporal behavioral models that track distributions of agent decisions over time can detect T2 anomalies (within-session behavioral shifts) and T3/T4a anomalies (cross-session drift). This requires maintaining behavioral baselines across sessions, an infrastructure requirement not met by current agent deployment frameworks. T4b threats (dormant backdoors embedded in model weights) are not detectable by this approach: the payload produces no behavioral anomaly until a specific trigger condition is met, at which point monitoring cannot distinguish the triggered response from intended behavior without knowledge of the trigger.

\subsubsection{The Defense-in-Depth Principle for Agents}

No single-layer defense is sufficient. Effective agentic security requires defense-in-depth: multiple independent controls at different layers, each covering the failure modes of the others. A minimal defense-in-depth stack for production agents:

\begin{enumerate}[leftmargin=1.5em,noitemsep]
  \item L1+L4: Safety-trained model + tool output filtering
  \item L3: Memory access control + write provenance tracking
  \item L4+L6: Tool sandboxing + MCP server authentication
  \item L5: Inter-agent message signing + behavioral correlation monitoring
  \item L7: Comprehensive action logging + mandatory human review for irreversible actions
\end{enumerate}

\begin{table*}[!t]
\centering
\footnotesize
\caption{Cross-layer defense taxonomy. Mechanism archetype (col.\ 2) is drawn from the \emph{a priori} list in \secref{sec:defense}. \dag: operational procedure rather than technical mechanism; listed for completeness.}
\label{tab:defense}
\renewcommand{\arraystretch}{1.0}
\setlength{\tabcolsep}{4pt}
\begin{tabularx}{\textwidth}{p{2.2cm}p{2.7cm}p{1.0cm}p{1.0cm}X}
\toprule
\textbf{Defense Class} & \textbf{Mechanism Archetype} & \textbf{$L$} & \textbf{$T$} & \textbf{Key Limitation} \\
\midrule
Safety Training          & Training-time alignment      & L1       & T1       & Fails vs.\ gradient-based / OOD attacks \\
Adversarial Fine-tuning  & Adversarial hardening        & L1       & T1       & Novel attack classes bypass training set \\
I/O Filtering            & I/O classification           & L1, L4   & T1       & Oblivious adversary assumed; semantic attacks pass \\
Prompt Injection Detect. & Pattern-based detection      & L4       & T1--T2   & High FPR; adaptive adversary evades known patterns \\
Tool Sandboxing          & Execution isolation          & L4, L6   & T1--T2   & No semantic injection prevention \\
Memory Access Control    & Access control (RBAC)        & L3       & T1--T3   & Per-principal key management overhead \\
Memory Consensus Valid.  & Multi-source consensus       & L3       & T1--T3   & ${\sim}15\%$ throughput; ${>}50\%$ honest sources required \\
Memory Invariant Enforc. & Invariant enforcement        & L3       & T3--T4   & Specification burden; unanticipated drift evades \\
Inter-Agent Auth.        & Cryptographic auth.          & L5       & T1--T2   & Compromised signer messages pass undetected \\
Collusion Detection      & Behavioral anomaly detection & L5       & T2--T3   & Known encodings only; novel steganography evades \\
MCP Code Signing         & Provenance verification      & L6       & T1       & Ecosystem adoption needed; no tool-behavior audit \\
ABOM Auditing            & ABOM enumeration             & L6       & T1--T2   & No standard; closed-model weight hashes inaccessible \\
Behavioral Monitoring    & Longitudinal monitoring      & L2, L7   & T3--T4   & Cross-session logs required; deployment-specific baselines \\
Human-in-the-Loop\dag    & Operational review           & L4, L7   & T1--T3   & Scalability limit; review quality degrades at volume \\
Interpretability Tools   & Mechanistic interpretability & L2, L7   & T1--T4   & Component-level only; multi-step trajectories unsupported \\
\bottomrule
\end{tabularx}
\end{table*}

\subsubsection{Defense Coverage vs. Attack Coverage}
\label{sec:defcov}

\figref{fig:defcov} renders the \emph{claimed} defense coverage by projecting each row of \tabref{tab:defense} onto the \lasm\ $\times$ temporality grid. Side-by-side with \figref{fig:temporality} (attack coverage), it makes a structural claim explicit: defense coverage is concentrated where attacks are well-studied (L1, L3, L4 $\times$ T1--T2), and the under-studied attack zone $\{L_5, L_6, L_7\} \times \{T_3, T_4\}$ is also the under-defended zone. Seven of 28 cells have zero defense coverage: $(L_1, T_2)$, $(L_1, T_3)$, $(L_1, T_4)$, $(L_4, T_4)$, $(L_5, T_4)$, $(L_6, T_3)$, and $(L_6, T_4)$. Three of these cells contain documented attacks (\tabref{tab:heatmap_counts}), an asymmetry strong enough to elevate to a numbered finding.

\begin{propositionbox}{Empirical Finding 3: Defense--Attack Coverage Asymmetry}
Of the 28 cells in the \lasm\ $\times$ temporality grid, 7 have zero defense coverage in the surveyed corpus. Three of these seven contain documented attacks: $(L_4, T_4)$, $(L_5, T_4)$, $(L_6, T_3)$. Defense research effort tracks attack research effort, not threat severity, so the under-defended zone coincides with the under-attacked zone of EF1. Defense-in-depth across the seven layers is an empirical engineering necessity rather than a generic principle. The action-log schema in \secref{sec:actionlog} and the ABOM specification in \secref{sec:ecosystem} close two of the seven gaps.
\end{propositionbox}

\textbf{Coverage caveats.} The figure shows \emph{claimed} coverage, not empirical effectiveness. Three caveats are essential. First, Interpretability Tools is listed as covering T1--T4 but its reported limitation (``insufficient for multi-step trajectories; component-level methods only'') means practical coverage is closer to T1. Second, several cells with non-zero coverage rely on a single defense class with a high false-positive cost (e.g., $(L_4, T_3)$ relies on Human-in-the-Loop, which scales poorly). Third, defenses listed under T4 coverage (Memory Invariant Enforcement, Behavioral Monitoring, Interpretability) all have unresolved specification burdens—none has been empirically validated against T4 attacks because no T3/T4 benchmark exists (Gap 1). The figure should therefore be read as an upper bound on defense coverage; the empirical picture is more pessimistic.

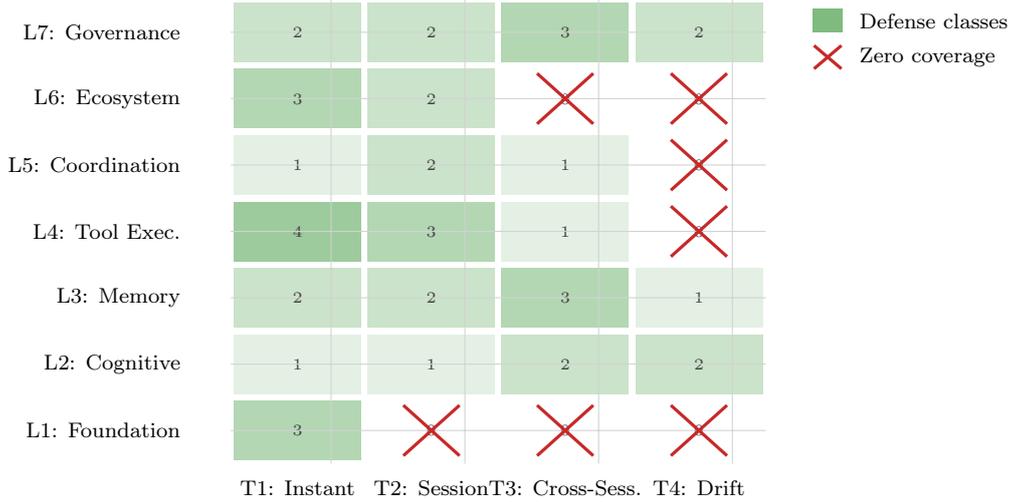
\begin{figure}[t]
\centering
\begin{tikzpicture}[scale=0.88, font=\small]
  \foreach \y/\lbl in {1/L1: Foundation, 2/L2: Cognitive, 3/L3: Memory,
                        4/L4: Tool Exec., 5/L5: Coordination,
                        6/L6: Ecosystem, 7/L7: Governance} {
    \node[anchor=east, font=\scriptsize] at (-0.1, \y) {\lbl};
  }
  \foreach \x/\lbl in {1.5/T1: Instant, 3.5/T2: Session,
                        5.5/T3: Cross-Sess., 7.5/T4: Drift} {
    \node[anchor=north, font=\scriptsize] at (\x, 0.40) {\lbl};
  }
  \fill[propborder!70,opacity=0.50] (0.55,0.55) rectangle (2.45,1.45);
  \fill[propborder!70,opacity=0.18] (0.55,1.55) rectangle (2.45,2.45);
  \fill[propborder!70,opacity=0.18] (2.55,1.55) rectangle (4.45,2.45);
  \fill[propborder!70,opacity=0.35] (4.55,1.55) rectangle (6.45,2.45);
  \fill[propborder!70,opacity=0.35] (6.55,1.55) rectangle (8.45,2.45);
  \fill[propborder!70,opacity=0.35] (0.55,2.55) rectangle (2.45,3.45);
  \fill[propborder!70,opacity=0.35] (2.55,2.55) rectangle (4.45,3.45);
  \fill[propborder!70,opacity=0.50] (4.55,2.55) rectangle (6.45,3.45);
  \fill[propborder!70,opacity=0.18] (6.55,2.55) rectangle (8.45,3.45);
  \fill[propborder!70,opacity=0.65] (0.55,3.55) rectangle (2.45,4.45);
  \fill[propborder!70,opacity=0.50] (2.55,3.55) rectangle (4.45,4.45);
  \fill[propborder!70,opacity=0.18] (4.55,3.55) rectangle (6.45,4.45);
  \fill[propborder!70,opacity=0.18] (0.55,4.55) rectangle (2.45,5.45);
  \fill[propborder!70,opacity=0.35] (2.55,4.55) rectangle (4.45,5.45);
  \fill[propborder!70,opacity=0.18] (4.55,4.55) rectangle (6.45,5.45);
  \fill[propborder!70,opacity=0.50] (0.55,5.55) rectangle (2.45,6.45);
  \fill[propborder!70,opacity=0.35] (2.55,5.55) rectangle (4.45,6.45);
  \fill[propborder!70,opacity=0.35] (0.55,6.55) rectangle (2.45,7.45);
  \fill[propborder!70,opacity=0.35] (2.55,6.55) rectangle (4.45,7.45);
  \fill[propborder!70,opacity=0.50] (4.55,6.55) rectangle (6.45,7.45);
  \fill[propborder!70,opacity=0.35] (6.55,6.55) rectangle (8.45,7.45);
  \foreach \x/\y/\cnt in {
    1.5/1/3, 3.5/1/0, 5.5/1/0, 7.5/1/0,
    1.5/2/1, 3.5/2/1, 5.5/2/2, 7.5/2/2,
    1.5/3/2, 3.5/3/2, 5.5/3/3, 7.5/3/1,
    1.5/4/4, 3.5/4/3, 5.5/4/1, 7.5/4/0,
    1.5/5/1, 3.5/5/2, 5.5/5/1, 7.5/5/0,
    1.5/6/3, 3.5/6/2, 5.5/6/0, 7.5/6/0,
    1.5/7/2, 3.5/7/2, 5.5/7/3, 7.5/7/2} {
    \node[font=\tiny, black!75] at (\x,\y) {\cnt};
  }
  \draw[gray!35] (0.5,0.5) grid[xstep=2,ystep=1] (8.5,7.5);
  \foreach \x/\y in {3.5/1, 5.5/1, 7.5/1, 7.5/4, 7.5/5, 5.5/6, 7.5/6} {
    \draw[alertred, very thick]
      ($(\x,\y)+(-0.42,-0.38)$) -- ($(\x,\y)+(0.42,0.38)$);
    \draw[alertred, very thick]
      ($(\x,\y)+(-0.42,0.38)$) -- ($(\x,\y)+(0.42,-0.38)$);
  }
  \fill[propborder!70,opacity=0.85] (9.2,7.0) rectangle (9.65,7.35);
  \node[anchor=west,font=\scriptsize] at (9.75,7.17) {Defense classes};
  \draw[alertred, very thick]
    (9.22,6.45) -- (9.63,6.80);
  \draw[alertred, very thick]
    (9.22,6.80) -- (9.63,6.45);
  \node[anchor=west,font=\scriptsize] at (9.75,6.62) {Zero coverage};
\end{tikzpicture}
\caption{Defense coverage heatmap (mirror to \figref{fig:temporality}). Cell value = number of defense classes from \tabref{tab:defense} that claim coverage of that $(L_i, T_k)$ cell; opacity scales linearly to the maximum cell ($L_4 \times T_1 = 4$). Red crosses mark cells with zero defense coverage. Cross-referencing with \tabref{tab:heatmap_counts}: $(L_4, T_4) = 2$ attack papers but $0$ defenses, $(L_5, T_4) = 1$ attack paper but $0$ defenses, $(L_6, T_3) = 1$ attack paper but $0$ defenses, these are the most acute under-coverage gaps. Coverage at $(L_7, T_4) = 2$ is overstated; see prose caveats.}
\label{fig:defcov}
\end{figure}


\subsubsection{Defense Recipes for Documented Attack Classes}
\label{sec:recipes}

For each canonical attack class, \tabref{tab:recipes} reports the minimal stack of defenses, assuming an adaptive adversary aware of deployed controls, that closes the dominant trust-inversion or temporal-escalation pathway. All listed controls are required jointly; residual risks name what each recipe does \emph{not} cover.

\begin{table*}[!t]
\centering
\footnotesize
\caption{Defense recipes for canonical attack classes. All listed controls in a recipe are required jointly; order follows runtime application.}
\label{tab:recipes}
\renewcommand{\arraystretch}{1.15}
\begin{tabularx}{\textwidth}{p{2.6cm}p{1.0cm}Xp{3.2cm}}
\toprule
\textbf{Attack class} & \textbf{(L,T)} & \textbf{Defense recipe} & \textbf{Residual risk} \\
\midrule
Indirect prompt injection (web/file/email) &
$L_4,\,T_1$--$T_2$ &
StruQ instruction--data separation~\cite{chen2024struq}; $\mathcal{R}_4$ output filter; tool sandbox~\cite{ruan2024sandbox}; action confirmation for irreversible ops. &
GCG-style low-perplexity payloads; T3 variants persisted into memory. \\
\midrule
Cross-session memory poisoning (PoisonedRAG) &
$L_3,\,T_3$ &
Memory RBAC; A-MemGuard consensus validation~\cite{memorySecurity2025}; write-time provenance; trust-scored retrieval ranking. &
Multi-source agreement on malicious entry; T4a drift. \\
\midrule
MCP supply-chain compromise (Postmark-class) &
$L_6,\,T_2$ &
ABOM hash verification at context init (\secref{sec:ecosystem}); MCP manifest signing; per-server sandbox; tool-call audit log (\secref{sec:actionlog}). &
Signed-but-compromised vendor; dynamic tool-description updates. \\
\midrule
Infectious jailbreak / trust chain &
$L_5,\,T_1$--$T_2$ &
Inter-agent message signing; heterogeneous model families; behavioral correlation monitoring; channel sandbox. &
Steganographic collusion ($T_2$--$T_3$); compromise via shared L3 memory. \\
\midrule
Sleeper agent / dormant backdoor &
$L_1$--$L_2,\,T_4$ &
Provenance-verified checkpoints; adversarial fine-tune audits; mechanistic interpretability on weights; L7 longitudinal monitoring. &
Triggers untouched in evaluation; weight steganography below interpretability resolution. \\
\midrule
Steganographic collusion &
$L_5,\,T_1$--$T_2$ &
Colosseum auditing~\cite{colosseum2025}; structured-output enforcement to lower $H(X)$ on inter-agent channels; heterogeneous models. &
Novel encodings outside training; info-theoretic lower bound (\secref{sec:multiagent}). \\
\midrule
Agentic insider threat (alignment failure) &
$L_7,\,T_{4a}$ &
Action-log schema (\secref{sec:actionlog}); deployment-specific behavioral baseline; L2/L7 drift detection; human review for irreversible ops. &
T4b dormant backdoors; drift detectors need session volume to converge. \\
\bottomrule
\end{tabularx}
\end{table*}

\subsection{Synthesis: From Defenses Back to the Research Questions}
\label{sec:def_synthesis}

The defense taxonomy and recipes answer RQ4 (\secref{sec:intro}) by binding each documented attack class to a layered control stack and naming the residual risks. Three structural conclusions tie this section back to the rest of the paper.

\emph{First, the defense surface is shaped by the attack surface, not by threat severity.} The defense coverage heatmap (\figref{fig:defcov}) and Empirical Finding~3 establish that defenses concentrate where attack research concentrates, and stop where it stops. This is the engineering analogue of Empirical Finding~1: the inverse correlation between severity and effort propagates from the attack literature into the defense literature.

\emph{Second, defense-in-depth is not optional.} Theorem~\ref{thm:nontrans} (Appendix~\ref{app:proof}) gives the formal reason: any single-layer defense has zero coverage against an $\mathcal{R}_j$-localized attack at a different layer. The seven zero-coverage cells in \figref{fig:defcov} are exactly the cells where this lower bound bites; a deployment with a strong $L_1$ control but no $L_4$ control is, in the worst case, no better defended at $L_4$ than a deployment with neither.

\emph{Third, the recipes rest on three primitives that the corpus does not yet provide.} The recipes in \tabref{tab:recipes} repeatedly invoke (a) cross-session memory auditing (Gap~1's prerequisite), (b) trust-marking on tool returns (no widely deployed implementation), and (c) action-log infrastructure (\secref{sec:actionlog}, currently a proposal). Closing the under-defended cells therefore requires \emph{infrastructure work}, not new ML research, which is why three of the five gaps in \secref{sec:gaps} are tractable on engineering rather than research timescales.

\subsection{Evaluation Frameworks and Benchmarks}
\label{sec:eval}

\subsubsection{Benchmark Landscape}

General-purpose LLM evaluation suites such as HELM~\cite{liang2022helm} and TruthfulQA~\cite{lin2022truthfulqa} measure accuracy and factuality but were not designed for adversarial security evaluation. \tabref{tab:benchmarks} surveys the current agent security benchmark landscape. Every benchmark in the table evaluates T1 or T2 threats. No security benchmark evaluates T3 or T4 threats, which means there is no way to measure whether defenses against slow-burn attacks improve over time. The three L7$\times$T4 entries in \tabref{tab:heatmap_counts} are alignment theory and governance papers, not evaluation frameworks, and do not contradict this finding.

\begin{table*}[!t]
\centering
\small
\caption{Agent security evaluation benchmarks, mapped to \lasm\ layers and temporality classes.}
\label{tab:benchmarks}
\renewcommand{\arraystretch}{1.1}
\begin{tabular}{p{2.5cm}p{1.4cm}p{1.2cm}p{0.8cm}p{6.5cm}}
\toprule
\textbf{Benchmark} & \textbf{$L$ Cov.} & \textbf{$T$ Cov.} & \textbf{Year} & \textbf{Focus / Key Limitation} \\
\midrule
HarmBench~\cite{harmbench2024} & L1 & T1 & 2024 & Standardized jailbreak red-teaming; no tool-use or multi-agent scenarios \\
AgentBench~\cite{agentBench2023} & L2--L4 & T1--T2 & 2023 & Multi-domain task completion; no adversarial evaluation \\
InjecAgent~\cite{injecagent2024} & L4 & T1 & 2024 & First indirect-injection benchmark for tool-use agents \\
AgentDojo~\cite{agentDojo2024} & L4 & T1--T2 & 2024 & Prompt injection in tool-use pipelines; single-session only \\
AgentSecBench~\cite{agentSecBench2024} & L2--L4 & T1--T2 & 2024 & Prompt injection and tool misuse; no multi-agent evaluation \\
ASB~\cite{agentSecBenchASB2025} & L2--L4 & T1--T2 & 2025 & Formalizes attack/defense pairs across 13 scenarios; no T3/T4 evaluation \\
AgentHarm~\cite{agentHarm2025} & L1--L4 & T1--T2 & 2025 & Multi-step harm robustness; single-agent, no memory evaluation \\
ARE~\cite{multimodalRobust2025} & L1, L4 & T1 & 2025 & Pixel-level adversarial robustness for visual agents \\
WIPI~\cite{wipi2024} & L4 & T1--T2 & 2024 & Web-page-mediated prompt injection against browsing agents \\
Colosseum~\cite{colosseum2025} & L5 & T1--T2 & 2025 & Collusion auditing in cooperative multi-agent systems \\
\bottomrule
\end{tabular}
\end{table*}

\begin{propositionbox}{Empirical Finding 4: T3/T4 Benchmark Vacuum}
Of the 10 surveyed agent-security benchmarks, none evaluates T3 (cross-session) or T4 (drift/dormant) threats. Every published benchmark terminates at T1 or T2. The consequence: defenses claiming T3/T4 coverage in \tabref{tab:defense} cannot be empirically distinguished from defenses claiming no such coverage; T3/T4 progress is structurally unmeasurable. This finding is the upstream constraint on Gap~1 (\secref{sec:gaps}) and an upper bound on every claim in the right-half columns of \figref{fig:temporality} and \figref{fig:defcov}.
\end{propositionbox}

\subsubsection{Methodological Issues in Existing Evaluations}

Beyond coverage gaps, we identify three methodological issues that limit the reliability of existing agent security evaluations:

\textbf{Oblivious adversary assumption.} Most evaluations test against a fixed attack strategy. Real adversaries adapt iteratively based on observed agent behavior. \emph{Adaptive attack} evaluation, where the adversary is allowed to query the defended system and refine their attack before evaluation, yields substantially different (typically worse) defense performance.

\textbf{No false positive reporting.} Security evaluations consistently report attack mitigation rates (true positive rate for defense) but not false positive rates (\ie\ rate at which the defense incorrectly flags legitimate agent behavior). A defense achieving 95\% attack mitigation at 40\% false positive rate is impractical; this information is absent from most published evaluations.

\textbf{Single-agent evaluation for multi-agent threats.} Benchmarks such as HarmBench and AgentHarm evaluate individual agents in isolation. They do not evaluate coordinated attacks against multi-agent systems, nor do they assess the spread of compromise through agent networks.

\subsubsection{Evaluation Recommendations}

We propose the following standards for rigorous agent security evaluation:

\begin{enumerate}[leftmargin=1.5em,noitemsep]
  \item Report results under both \emph{oblivious} and \emph{adaptive} adversary settings.
  \item For T3 evaluation, use multi-session protocols with $\geq$5 session transitions between injection and exploitation.
  \item Report false positive rates alongside true positive rates for all defenses.
  \item Evaluate under heterogeneous agent deployments (mixed model families) for L5 threats.
  \item Standardize threat model descriptions using the four-dimensional characterization from \secref{sec:background}.
\end{enumerate}


\section{Open Problems and Future Directions}
\label{sec:open}
\label{sec:gaps}

\subsection{Research Gaps}

The \lasm\ $\times$ temporality heatmap (\figref{fig:temporality}) locates five research gaps in the high-risk, low-coverage zone, summarized in \tabref{tab:gaps}. Each entry names the prerequisite that has to land before the gap is measurable, the methodological tradition the work draws on, and the difficulty class. \secref{sec:future} then proposes how progress is made on each.

\begin{table}[h]
\centering
\footnotesize
\caption{Five research gaps in the high-risk, low-coverage zone of the \lasm\,$\times$\,temporality heatmap. ``Difficulty'' classes: \textit{Std} = standards/coordination, \textit{Eng} = engineering composition, \textit{Multi} = multi-disciplinary, \textit{Fund} = fundamental basic research.}
\label{tab:gaps}
\renewcommand{\arraystretch}{1.15}
\begin{tabularx}{\linewidth}{cp{3.0cm}cp{1.1cm}X}
\toprule
\textbf{\#} & \textbf{Gap} & \textbf{Cell} & \textbf{Diff.} & \textbf{What unlocks it} \\
\midrule
1 & Cross-session attack benchmarks & L3$\times$T3 & Eng & Multi-session replay; memory serialization; ground-truth labels; metrics scoring future-session behavioral change rather than injection success. Prerequisite for measuring all T3 defenses in \tabref{tab:defense}. \\
2 & Emergent misalignment detection & L7$\times$T4 & Fund & Statistical process control on high-dimensional stochastic behavioral distributions; longitudinal agent datasets; causal attribution distinguishing drift from legitimate adaptation. The L7$\times$T4 cell has zero defense coverage today. \\
3 & General steganographic-collusion detection & L5$\times$T1--T2 & Fund & A detector for novel encoding schemes (Colosseum~\cite{colosseum2025} handles only known schemes). Reduces to detecting statistical anomalies in LLM output distributions against adaptive adversaries. \\
4 & ABOM standardization and MCP certification & L6$\times$T1--T2 & Std & Specification working group (NTIA-SBOM analogue); automated ABOM tooling; mandatory-disclosure MCP certification regime. The only fully solvable-now node. \\
5 & System-level accountability frameworks & L7$\times$T1--T4 & Multi & Trajectory-level causal attribution; formal verification for agent plans; liability allocation across the accountability chain (\secref{sec:governance}). Without it, high-stakes deployment is defensible only by restricting agents to fully reversible actions~\cite{agentAccountability2025}. \\
\bottomrule
\end{tabularx}
\end{table}

\subsection{Future Directions}
\label{sec:future}

The five gaps in \secref{sec:gaps} identify \emph{what is missing}; this section proposes \emph{how progress is made on each}, plus directions that extend beyond the gap set. The mapping is summarized in \tabref{tab:future_map}: each future direction is anchored to one or more gaps, and labeled by the methodological tradition it draws on. This separates Gaps (problem statements) from Future Directions (proposed approaches), eliminating the redundancy between the two sections in earlier drafts.

\paragraph{Gap dependency structure.} The five gaps in \tabref{tab:gaps} are not independent: Gap~1 (T3/T4 benchmarks) is the most upstream prerequisite, unlocking measurement for Gaps~2 and~3; Gap~4 (ABOM/MCP certification) is the only fully-solvable-now node; and Gap~5 (accountability frameworks) is a parallel-track prerequisite for Gap~2's runtime side. Suggested sequencing: \textbf{Gap~1 first}, then in parallel \textbf{Gap~4} (engineering coordination) and \textbf{Gap~5} (governance scaffolding), with Gaps~2 and~3 unlocked by those investments.

\begin{table}[h]
\centering
\footnotesize
\caption{Future directions mapped to research gaps and methodological traditions.}
\label{tab:future_map}
\renewcommand{\arraystretch}{1.1}
\begin{tabularx}{\linewidth}{p{3.4cm}p{1.8cm}X}
\toprule
\textbf{Future direction} & \textbf{Targets} & \textbf{Tradition} \\
\midrule
Formal security models for agentic protocols & Gaps 2, 3, 5 & Cryptography; formal methods \\
Runtime behavioral attestation & Gap 5 & Trusted computing; TEEs \\
Temporal anomaly detection infrastructure & Gaps 1, 2 & Statistical process control; distributed systems \\
Multi-stakeholder governance frameworks & Gaps 4, 5 & Law; policy; standards \\
Defense-parity autonomous agents & Beyond Gaps & ML systems; dual-use research \\
\bottomrule
\end{tabularx}
\end{table}

\subsubsection{Formal Security Models for Agentic Protocols (Gaps 2, 3, 5)}

Cryptographic protocols achieve provable security via formal models (Dolev-Yao, universal composability). Agentic AI lacks analogous formal security models. Developing such models requires: (a) a formal language for specifying agent behavior and security properties; (b) a model of adversary capabilities in the agentic context; and (c) composition theorems allowing security proofs for complex multi-agent systems to be built from component proofs. A formal model that captures principal trust inversion (\secref{sec:background}) and proves invariants against it would address Gap 5 directly and provide the theoretical machinery needed for Gaps 2 and 3.

\subsubsection{Runtime Behavioral Attestation (Gap 5)}

A trusted execution environment for agent actions, providing cryptographic attestation that actions were sanctioned by a verified principal hierarchy, would address the accountability gap by binding each entry of the action-log schema in \secref{sec:actionlog} to an attested execution context. Hardware trusted execution environments (Intel TDX, AMD SEV) provide the computational substrate; adapting them to the inference stack of LLM agents (stochastic, large memory footprint, streaming outputs) is a significant engineering challenge.

\subsubsection{Temporal Anomaly Detection Infrastructure (Gaps 1, 2)}

Addressing T3 and T4 threats requires building infrastructure that agents currently lack: (a) persistent, tamper-evident behavioral logs spanning multiple sessions (the schema in \secref{sec:actionlog} is a starting point); (b) statistical baselines for normal agent behavior specific to each deployment context; and (c) real-time anomaly scoring that can distinguish malicious behavioral shifts from legitimate adaptation. The benchmark infrastructure called for in Gap 1 is a prerequisite: without ground-truth T3/T4 attack-defense pairs, anomaly scorers cannot be calibrated.

\subsubsection{Multi-Stakeholder Governance Frameworks (Gaps 4, 5)}

The distributed accountability chain (\secref{sec:governance}) requires governance frameworks spanning developers, operators, users, and regulators. Legal frameworks assigning clear liability for agentic harms, analogous to product liability for physical goods, are an urgent policy priority as agentic AI deployments reach the scale at which harmful incidents become statistically inevitable. ABOM standardization (Gap 4) and a certification regime for MCP servers are two candidate near-term policy instruments that translate the schema in \secref{sec:actionlog} and \secref{sec:ecosystem} into enforceable obligations.

\subsubsection{Defense-Parity Autonomous Agents (beyond the gap set)}

LLM agent teams can already autonomously exploit zero-day software vulnerabilities~\cite{kang2024teams}. No analogous defensive capability exists: no agent system can autonomously defend against comparable threats. Developing defensive agents that autonomously detect, analyze, and mitigate attacks against themselves and their deployment environment is a long-term research objective requiring community coordination. This direction is not anchored to a single gap because parity is a moving target as offensive capabilities scale; it is included to acknowledge that even fully closing Gaps 1--5 would not resolve the offense-defense asymmetry that emerges as agents become first-class actors in security pipelines.

\subsection{Comparison with Existing Surveys}
\label{sec:related}

\tabref{tab:survey_compare} compares this survey to ten prior surveys along seven structural criteria. The pattern is uniform: prior surveys offer either taxonomies without an architectural layer model, or single-layer deep dives without a temporality axis, or attack catalogues without formal claims. None combines all three.

\begin{table*}[!t]
\centering
\scriptsize
\caption{Comparison with related surveys. ``Layer depth'' counts architectural layers explicitly distinguished. ``Temp.\ classes'' counts formal temporality classes. ``Form.\ prop.'' is yes if the paper states at least one formal claim with proof or proof sketch. ``Public coding'' is yes if the paper-cell coding rationale is released as an appendix-grade artifact.}
\label{tab:survey_compare}
\renewcommand{\arraystretch}{1.1}
\setlength{\tabcolsep}{4pt}
\begin{tabularx}{\textwidth}{Xrrrrcccc}
\toprule
\textbf{Survey} & \textbf{\#Papers} & \textbf{Layer} & \textbf{Temp.} & \textbf{Cross-layer} & \textbf{Mem.} & \textbf{MCP/} & \textbf{Form.} & \textbf{Public} \\
 & & \textbf{depth} & \textbf{classes} & \textbf{figs.} & \textbf{Sec.} & \textbf{SC} & \textbf{prop.} & \textbf{coding} \\
\midrule
Wang \etal~\cite{agentsUnderThreat2025} & $\sim$80 & 0 & 0 & 0 & $\circ$ & --- & no & no \\
Chhabra \etal~\cite{agentSecSurvey2025a} & $\sim$120 & 0 & 0 & 0 & $\circ$ & $\circ$ & no & no \\
Tang \etal~\cite{agentSecSurvey2025b} & $\sim$110 & 3 & 0 & 0 & $\circ$ & --- & no & no \\
Lin \etal~\cite{memorySecurity2025} & $\sim$60 & 1 (L3) & informal & 0 & $\checkmark$ & --- & no & no \\
Hou \etal~\cite{mcp2025} & $\sim$40 & 1 (L6) & 0 & 0 & --- & $\checkmark$ & no & no \\
MITRE ATLAS~\cite{mitreAtlas2024} & ongoing TTPs & 0 & 0 & 0 & --- & $\circ$ & no & yes \\
He \etal~\cite{he2025emerged} (CSUR'25) & $\sim$120 & 0 & 0 & 0 & $\circ$ & $\circ$ & no & no \\
Raza \etal~\cite{trism2025} (TRiSM) & $\sim$60 & 0 & 0 & 0 & $\circ$ & $\circ$ & no & no \\
Yu \etal~\cite{trustworthyLLMAgentsKDD2025} (KDD'25) & $\sim$130 & 0 & 0 & 0 & $\circ$ & $\circ$ & no & no \\
Kim \etal~\cite{attackDefenseLandscape2026} (USENIX'26) & $\sim$140 & 3 & 0 & 0 & $\circ$ & $\circ$ & no & no \\
\midrule
\textbf{This survey} & \textbf{116} & \textbf{7} & \textbf{4 (T1--T4)} & \textbf{1 (Fig.\ \ref{fig:killchain})} & $\checkmark$ & $\checkmark$ & \textbf{yes} & \textbf{yes} \\
\bottomrule
\end{tabularx}
\end{table*}

Prior surveys cluster into four groups. \textbf{Type-centric surveys}~\cite{agentsUnderThreat2025,agentSecSurvey2025a,he2025emerged,trism2025,trustworthyLLMAgentsKDD2025} organize threats by attack type and track adversary intent, but do not map to architectural components, so defense placement cannot be read off them; Yu \etal~\cite{trustworthyLLMAgentsKDD2025} (KDD 2025) is the most recent example with the widest corpus ($\sim$130 papers) but presents no architectural layer model or temporality classification. \textbf{Single-layer deep dives}~\cite{memorySecurity2025,mcp2025} treat one \lasm\ layer in depth but say nothing about cross-layer interaction or temporality. The closest structural sibling, Tang~\etal~\cite{agentSecSurvey2025b}, splits attacks into External / Internal / Multi-Agent paradigms but conflates L3 (memory) with L4 (tool) under ``Internal'' and does not separate L6 from L4. Kim~\etal~\cite{attackDefenseLandscape2026} (USENIX Security 2026) is the most comprehensive attack-defense landscape study to date ($\sim$140 papers); it introduces a three-tier structural model (data, model, system) but does not formalize layer boundaries, offers no temporality axis, and does not cover L7 governance. \textbf{TTP-centric frameworks} (MITRE ATLAS~\cite{mitreAtlas2024}, OWASP LLM Top 10~\cite{owaspLLM2024}) answer ``what is the attacker doing'' rather than ``where is the control point''; they are complementary to \lasm\ but add no temporal axis and do not address the L5--L7 high-layer zone. Papernot~\etal~\cite{papernot2018sok} and Biggio--Roli~\cite{biggio2018wild} built the classical-ML threat taxonomy this survey extends to the agentic setting. No prior work combines a seven-layer framework with formal boundaries, a temporality classification, a cross-layer defense taxonomy with temporal coverage, and the empirical $\{L_5,L_6,L_7\} \times \{T_3, T_4\}$ under-studied zone.

\section{Conclusion}
\label{sec:conclusion}

AI agents now manage enterprise files, execute financial transactions, write and deploy code, and orchestrate multi-step workflows at scale. Their security is a societal question, not only a technical one. \lasm\ provides the structural primitive that stateless LLM safety and classical software security techniques each lack: a seven-layer decomposition with disjoint trust boundaries (Theorem~\ref{thm:nontrans}) where defense placement decisions follow from layer identity. \emph{Attack temporality} is the orthogonal axis: the highest-impact threats are not instantaneous attacks at well-studied layers but slow-burning attacks at system-level layers (L5--L7) that no existing benchmark measures. The 116-paper review confirmed this as an empirical finding: the $\{L_5, L_6, L_7\} \times \{T_3, T_4\}$ zone holds 6.3\% of all 144 paper-cell assignments but the highest documented real-world impact, with seven cells at zero defense coverage.

Of the five identified gaps, Gaps~1, 4, and~5 are near-term tractable: multi-session benchmark infrastructure (Gap~1), ABOM/MCP certification (Gap~4), and system-level accountability frameworks (Gap~5) are coordination and engineering problems whose components exist. Gaps~2 and~3 carry fundamental information-theoretic lower bounds and require sustained basic research. The field will likely follow the supply-chain security trajectory of 2010--2020: incident-driven T1/T2 coverage in the first half, T3/T4 infrastructure only if Gap~1's benchmarks and Gap~4's standards bodies are established. Agent security is a distributed systems problem in an adversarial ecosystem; treating it as a model problem or a prompt problem misallocates controls and leaves the highest-impact attack surface uninstrumented.

\section*{Declarations}

\paragraph{Competing interests} The author declares no competing interests.

{\sloppy
\bibliographystyle{vancouver}
\bibliography{references}
}

\appendix
\section{Limitations and Reproducibility}
\label{app:limitations}

\tabref{tab:limitations} consolidates the limitations of this survey by severity; narrative discussion of the underlying biases is in \secref{sec:caveats}. None of the limitations changes the qualitative findings: the under-studied zone is sparsely populated, defense coverage is concentrated where attack research is concentrated, and the five identified gaps are independently substantiated by the robustness analysis in Appendix~\ref{app:robustness}.

\begin{table*}[h]
\centering
\footnotesize
\caption{Consolidated limitations. \textbf{H}: high severity; \textbf{M}: medium; \textbf{L}: low.}
\label{tab:limitations}
\renewcommand{\arraystretch}{1.1}
\begin{tabularx}{\textwidth}{p{4.0cm}cXp{2.5cm}}
\toprule
\textbf{Limitation} & \textbf{Sev.} & \textbf{Mitigation already in place} & \textbf{See} \\
\midrule
Single-coder cell assignment & M & Operational coding rules + 4 worked examples; empirical robustness across 4 sensitivity protocols places EF1 within [4.81\%, 8.33\%] & App.\,\ref{app:coding}, \ref{app:robustness} \\
Headline analysis cutoff April 2025 & L & Post-cutoff scan (7 new entries) confirms all four Empirical Findings & \secref{sec:postcutoff} \\
Search-term and database bias & M & Forward/backward citation tracing from 3 seed papers & \secref{sec:caveats} \\
\lasm\ assumes a fixed deployment architecture & M & Layer applicability is per-deployment; explicit guidance given & \secref{sec:caveats} \\
MCP statistics from preprint~\cite{mcp2025} & L & Cross-referenced to documented Postmark incident & \secref{sec:ecosystem} \\
ABOM is a conceptual specification, not a ratified standard & L & Flagged as Gap~4; field schema in body & \secref{sec:ecosystem} \\
Defense corpus is published-only & L & Heatmap framed as upper bound on published defenses & \secref{sec:caveats} \\
Action-log schema is a proposal, not validated & L & Flagged as future work & \secref{sec:actionlog} \\
Information-theoretic steganography bound is informal & L & Stated as motivation for Gap~3 framing, not as theorem & \secref{sec:multiagent} \\
\bottomrule
\end{tabularx}
\end{table*}

The 116-paper corpus, per-paper coding rationale, and cell-count derivation are released as supplementary material (\path{supplementary/}). Key artefacts: \path{corpus_coding.csv} (144 paper-cell rows); \path{verify_coding.py} (one-command heatmap re-derivation); \path{robustness_analysis.py} (four sensitivity protocols on EF1); \path{abom/} (JSON Schema, Python validator, two example documents, 12-test suite).

\section{Coding Rules and Worked Examples}
\label{app:coding}

This appendix documents the operational rules used to assign each retained paper to one or more $(L_i, T_k)$ cells. The procedure is the formalization of the prose criteria in \secref{sec:methodology}.

\subsection{Layer-Assignment Rules}

A paper is assigned to layer $L_i$ if its primary technical contribution operates on the attack surface representation $\mathcal{R}_i$:

\begin{itemize}[leftmargin=1.5em,noitemsep]
\item \textbf{L1} if the contribution targets weight-space artifacts: gradient-based perturbations, weight extraction, backdoor poisoning of fine-tuning data, or alignment training. Test: \emph{would the attack/defense still apply if the in-context reasoning chain were replaced by a different decoding strategy?} If yes, L1.
\item \textbf{L2} if the contribution targets in-context reasoning artifacts: planning trajectories, chain-of-thought, tree-of-thought search, reflection. Test: \emph{does the attack/defense require access to the live reasoning state at inference time?} If yes, L2.
\item \textbf{L3} if the contribution involves persistent or session-scoped memory: RAG, vector stores, episodic/semantic memory, memory consolidation. Test: \emph{is there a write step whose effect is realized at a later read step?}
\item \textbf{L4} if the contribution targets tool I/O: indirect prompt injection via tool outputs, capability creep, autonomous exploitation, sandboxing. Test: \emph{does the payload arrive via a tool call return or affect a tool invocation?}
\item \textbf{L5} if the contribution requires $\geq 2$ communicating agents: trust chain attacks, infectious jailbreak, collusion, BFT. Test: \emph{is the attack/defense definable on a single agent in isolation?} If no, L5.
\item \textbf{L6} if the contribution targets deployment-time artifacts: MCP servers, model checkpoints, package supply chain, prompt template repositories. Test: \emph{is the compromise installed before the agent runtime begins?}
\item \textbf{L7} if the contribution targets accountability or observability: logging schemas, attribution methods, regulatory analysis, alignment-failure detection without an external adversary.
\end{itemize}

Multi-layer assignment is permitted: a paper that demonstrates a kill chain (\eg\ L4 injection $\to$ L3 memory poisoning) is coded in both cells. The 144 paper-cell assignments from 116 papers reflect this; the mean of 1.24 cells per paper is dominated by 26 papers coded in 2 cells and 2 papers coded in 3 cells.

\subsection{Temporality-Assignment Rules}

Temporality assignment follows the structural definitions in \secref{sec:temporality}, applied via these tests:

\begin{itemize}[leftmargin=1.5em,noitemsep]
\item \textbf{T1} if the entire attack-effect chain completes within one inference call. Default for jailbreaks, GCG-style suffixes, single-turn injection.
\item \textbf{T2} if the payload persists across turns within one session and the harmful effect is bounded by the session.
\item \textbf{T3} \emph{only if} the paper explicitly involves a memory-layer write whose payload activates in a future session. The cross-session boundary must be explicit; mere multi-turn-within-session does not qualify.
\item \textbf{T4} if (T4a) the harmful behavior is a drift property with no discrete payload, or (T4b) the payload resides in model weights or training data and is trigger-conditioned without a session-layer write.
\end{itemize}

When a paper studies multiple temporal patterns, it is coded under each. For ambiguous cases, the conservative rule in \secref{sec:methodology} applies (more restrictive / faster temporality). For papers that propose both attack and defense (\eg\ Lin \etal~\cite{memorySecurity2025}), the layer is coded once but the temporality is coded by the union of attack and defense temporalities.

\subsection{Worked Examples}

\textbf{Example 1: Niu \etal~\cite{infectiousJailbreak2024} (Infectious Jailbreak).}
Layer test: requires $\geq 2$ communicating agents (the propagation graph is essential to the contribution) $\to$ L5. Secondary L1 effect (jailbreak override at each relay) is mentioned but is not the primary contribution; not separately coded. Temporality: propagation occurs within and across turns of an ongoing multi-agent session $\to$ T1--T2. \textbf{Final cells: $(L_5, T_1)$, $(L_5, T_2)$.}

\textbf{Example 2: Hubinger \etal~\cite{sleepingAgents2024} (Sleeper Agents).}
Layer test: payload is installed via fine-tuning (training-time), trigger fires at inference; both $\mathcal{R}_1$ (weight artifact) and $\mathcal{R}_2$ (trigger-conditioned reasoning) are exercised $\to$ L1 \emph{and} L2. Also has L7 implications (governance survey of detection methods) but those are framed as discussion, not contribution. Temporality: trigger-conditioned, no memory write $\to$ T4b. Cross-session weight persistence also produces T3 effects in the cited follow-up evaluations. \textbf{Final cells: $(L_1, T_3)$, $(L_1, T_4)$, $(L_2, T_4)$, $(L_7, T_4)$.}

Two further worked examples (PoisonedRAG, MCP security survey) appear in \path{supplementary/CODING_EXAMPLES.md}.

\subsection{Sensitivity Analysis}

To bound the effect of single-coder ambiguity, we re-coded the 8 papers with the lowest-confidence assignments (those where the layer or temporality was contested under the conservative rule). Re-coding these papers under the alternative interpretation produces the following changes to \tabref{tab:heatmap_counts}:

\begin{itemize}[leftmargin=1.5em,noitemsep]
\item Largest single-cell change: $(L_2, T_3)$ shifts by $\pm 2$ depending on whether reward-hacking papers~\cite{rewardHacking2023,skalse2022reward} are coded as cross-session (T3) or as trigger-conditioned (T4b).
\item The under-studied zone $\mathcal{U} = \{L_5, L_6, L_7\} \times \{T_3, T_4\}$ count of 8 changes by at most $\pm 2$ under any combination of the 8 alternative codings, remaining well below 15\% of all paper-cell assignments.
\item No re-coding scenario shifts a cell from the $\geq 5$-paper region into the under-studied zone or vice versa.
\end{itemize}

The under-coverage finding (Empirical Finding 1) is therefore robust to single-coder ambiguity within the bounds of these 8 papers. The supplementary CSV identifies these 8 papers and the alternative codings considered.

\section{Defense Non-Transferability Theorem}
\label{app:proof}

We state Proposition 1 formally and give a proof sketch. The full proof, two corollaries, and discussion of boundary conditions appear in \path{supplementary/PROOF_THEOREM1.md}.

A defense $d \in \mathcal{D}_i$ is \emph{pure} if its alert is a deterministic function of inputs from $\mathcal{R}_i$ alone, with no side-channel access to $X_j$ for any $j \neq i$. An attack $a \in \mathcal{T}_j$ is \emph{$\mathcal{R}_j$-localized} if its payload appears only in $\mathcal{R}_j$, leaving every $X_k$ ($k \neq j$) unchanged in distribution.

\begin{theorem}[Defense Non-Transferability]
\label{thm:nontrans}
Let $i \neq j$, $d$ a pure layer-$i$ defense, and $a$ an $\mathcal{R}_j$-localized attack. Then $\Pr[d \text{ alerts under } a] = \Pr[d \text{ alerts under no attack}]$. The defense $d$ has zero detection power against $a$.
\end{theorem}

\begin{proof}[Proof sketch]
$d$ depends on its input only through the projection $\pi_i$, and $X_i \mid a \stackrel{d}{=} X_i$ by $\mathcal{R}_j$-localization. The alert distributions under $a$ and under no attack are identical. \qed
\end{proof}

The result has two consequences. A rational adversary will route the payload through whichever $\mathcal{R}_j$ has no defense, so every uninstrumented layer is exploitable in the worst case. Cross-layer ensembles obey the same bound: an ensemble that fails to instrument $\mathcal{R}_j$ has zero coverage against $\mathcal{R}_j$-localized attacks. The seven zero-coverage cells in \figref{fig:defcov} are the empirical instances where this bound bites.

The theorem does not address impure defenses with multi-representation access, cross-layer attacks whose payload spans $\mathcal{R}_i \cap \mathcal{R}_j$ (which would violate disjointness), or agent-mediated cross-layer effects such as RLHF training shifting the marginal $\mathcal{R}_2$ distribution. The last case explains why safety fine-tuning fails against trigger-conditioned sleeper-agent attacks~\cite{sleepingAgents2024}: the fine-tuning is an $\mathcal{R}_1$ operation that cannot target a specific in-context $\mathcal{R}_2$ pattern.

\section{Robustness Analysis of Empirical Finding 1}
\label{app:robustness}

We test EF1 against four perturbations of the corpus coding, all reproducible from \path{supplementary/robustness_analysis.py}. \tabref{tab:robustness} summarizes the result. The under-studied zone fraction stays below 9\% across all $10^4$-iteration runs, far from any threshold that would falsify the qualitative claim.

\begin{table}[h]
\centering
\footnotesize
\caption{Robustness of EF1 under four perturbations. Baseline 6.3\% (9/144). Full method description in supplementary.}
\label{tab:robustness}
\renewcommand{\arraystretch}{1.1}
\begin{tabularx}{\linewidth}{lXr}
\toprule
\textbf{Perturbation} & \textbf{Range / 95\% CI} & \textbf{Stable?} \\
\midrule
Leave-one-paper-out ($n{=}116$) & 5.56--6.94\% & yes \\
Bootstrap, 10\% paper drop, $10^4$ runs & [4.81\%, 7.48\%] & yes \\
Cell perturbation, $p{=}0.10$, $10^4$ runs & [5.83\%, 8.33\%] & yes \\
Role ablation, drop one role & 4.76\% (drop survey) to 7.62\% (drop context) & yes \\
\bottomrule
\end{tabularx}
\end{table}

The strongest single-paper effect is $\pm 0.8$ percentage points. Random subset and cell perturbation inflate the fraction by at most 1.7 pp at the 95\% percentile. Dropping all surveys lowers the fraction to 4.76\%, still consistent with a sparse zone. None of these capture systematic coding bias; that limitation remains, and is partly addressed by the rules and worked examples in Appendix~\ref{app:coding}.

\section{OWASP/ATLAS vs.\ \lasm\ Defense Placement}
\label{app:owasp}

\tabref{tab:owasp_compare} demonstrates that ``Prompt Injection'' (LLM01) maps to three different \lasm\ layers (L1, L3, L5) each requiring layer-specific defenses, and that ``Insecure Plugin Design'' (LLM07) maps to L4 and L6 with non-overlapping controls. A designer following only OWASP/ATLAS would apply input filtering and package-signing across all five cases; \lasm\ shows that four of the five require additional controls that input-filtering and package-signing cannot provide.

\begin{table*}[h]
\centering
\footnotesize
\caption{Defense placement under type-centric taxonomy vs.\ \lasm\ for five documented attacks. Attacks 1--2 share OWASP category LLM07 but require different defenses; Attacks 3--5 share LLM01 but map to three distinct \lasm\ layers. OWASP assignments follow the closest-fit interpretation of the OWASP LLM Top 10 (2023); where the official documentation does not enumerate a specific attack the assignment is the authors' interpretation.}
\label{tab:owasp_compare}
\renewcommand{\arraystretch}{1.15}
\begin{tabularx}{\textwidth}{p{2.6cm}p{1.7cm}p{2.6cm}p{1.0cm}p{2.8cm}X}
\toprule
\textbf{Attack instance} & \textbf{OWASP LLM} & \textbf{Defense from type label} & \textbf{\lasm\ layer} & \textbf{\lasm-guided defense} & \textbf{Gap closed by \lasm} \\
\midrule
Postmark MCP BCC injection~\cite{postmarkMCP2025} &
LLM07: Insecure Plugin Design &
Audit npm deps; validate tool schemas &
L6 &
ABOM manifest hash verification before context init &
Schema validation misses BCC injection in natural-language tool descriptions; L4 output filtering misses legitimately formatted exfiltration calls. \\[2pt]

GitHub Copilot RCE via project files~\cite{agentInjectionCoding2025} &
LLM07: Insecure Plugin Design &
Audit npm deps; sign packages &
L4 &
Tool output sandboxing; principal trust marking on file-read returns &
Package signing prevents the malicious file's presence but does not contain execution once present; only L4 sandboxing limits blast radius if L6 fails. \\[2pt]

PoisonedRAG cross-session poisoning~\cite{poisonRAG2024} &
LLM03: Training Data Poisoning &
Data validation at ingestion; source verification &
L3 (T3) &
Memory consensus validation; write-time provenance tracking &
LLM03 mitigations validate the initial knowledge base, not runtime memory writes during operation. The required defense is write-time provenance, not ingestion-time validation. \\[2pt]

Infectious jailbreak via agent relay~\cite{infectiousJailbreak2024} &
LLM01: Prompt Injection &
Per-agent safety training; content classifiers &
L5 (T2) &
Inter-agent authentication; heterogeneous model families across agents &
A compromised agent relays a syntactically legitimate payload through its normal interface; per-agent classifiers cannot intercept it. \\[2pt]

GCG adversarial suffix~\cite{gcg2023} &
LLM01: Prompt Injection &
Perplexity-based input filter &
L1 &
Adversarial training on gradient-space perturbations; certified robustness &
GCG suffixes are constructed to maintain low perplexity, defeating filters; the required $\mathcal{R}_1$ defense is gradient-space adversarial training, not input-text filtering. \\
\bottomrule
\end{tabularx}
\end{table*}

\section{Action-Log Schema}
\label{app:actionlog}

A minimal action-log schema sufficient for trajectory attribution at L7. The fields below are not a contribution claim; they are presented to make Gap~5 (\tabref{tab:gaps}) concrete enough for engineering follow-up, and to anchor the recipes in \tabref{tab:recipes} that invoke \texttt{action\_log} as a control. If standardized analogously to syslog/RFC~5424 or OpenTelemetry traces, the schema would close the lowest-cost subset of the L7 gap.

\begin{propositionbox}{Action-Log Record (per agent step)}
\textbf{Identity \& trajectory.} \texttt{trace\_id} (UUIDv7, monotonic), \texttt{session\_id}, \texttt{step\_index}, \texttt{parent\_step\_id} (for branched plans). \\[2pt]
\textbf{Principal chain.} \texttt{developer\_id}, \texttt{operator\_id}, \texttt{user\_id}, \texttt{agent\_id}, \texttt{abom\_hash} (binds the trajectory to a deployment configuration; cf.\ \secref{sec:ecosystem}). \\[2pt]
\textbf{Decision context.} \texttt{prompt\_hash} (SHA-256 of full prompt at step), \texttt{memory\_reads} (list of $\langle\text{memory\_id}, \text{content\_hash}\rangle$), \texttt{tool\_outputs\_consumed} (list of $\langle\text{tool\_id}, \text{output\_hash}, \text{trust\_label}\rangle$). \\[2pt]
\textbf{Action \& effect.} \texttt{action\_type} $\in$ \{tool\_call, msg\_send, mem\_write, plan\_revise, terminate\}, \texttt{action\_payload\_hash}, \texttt{reversibility} $\in$ \{revertible, costly, irrevocable\}, \texttt{authorization\_path} (which principal sanctioned this action class). \\[2pt]
\textbf{Reasoning attestation (optional).} \texttt{cot\_hash} (SHA-256 of CoT trace, retained off-log for privacy if needed), \texttt{reasoning\_was\_visible\_to\_caller} (bool). \\[2pt]
\textbf{Integrity.} \texttt{record\_signature} (HMAC or transparency-log inclusion proof binding the record to an append-only ledger).
\end{propositionbox}

The schema enables trajectory reconstruction (via \texttt{parent\_step\_id} chains), action-class authorization audit (\texttt{authorization\_path} $\times$ \texttt{reversibility}), cross-session memory correlation (\texttt{memory\_reads} hashes matched to L3 write events), and supply-chain attribution (\texttt{abom\_hash}). All four governance gaps in \secref{sec:governance} presuppose records of approximately this granularity; without them, even a willing regulator could not enforce the requirement.


\end{document}